%% file: arxiv-main.tex
\documentclass[11pt,letterpaper]{article}

\usepackage[margin=1in]{geometry}

\usepackage[utf8]{inputenc}
\usepackage[T1]{fontenc}
\usepackage[colorlinks=true, linkcolor=blue, citecolor=green!60!black]{hyperref}
\usepackage{url}
\usepackage{booktabs}
\usepackage{amsfonts}
\usepackage{microtype}
\usepackage{xcolor}
\usepackage{nicefrac}
\usepackage{amssymb,amsthm,amsmath,wrapfig,dsfont,authblk,mathrsfs,enumerate}
\usepackage{natbib}
\usepackage{parskip}
\usepackage{indentfirst}
\usepackage{multirow}
\usepackage[nameinlink,noabbrev,capitalize]{cleveref}
\usepackage[ruled,noend]{algorithm2e}
\usepackage[framemethod=tikz]{mdframed}
\usepackage{thmtools,thm-restate}

\SetAlFnt{\small}
\SetAlCapFnt{\small}
\SetAlCapNameFnt{\small}
\SetAlCapHSkip{0pt}
\IncMargin{-\parindent}
\newtheorem{theorem}{Theorem}[section]
\newtheorem{lemma}[theorem]{Lemma}

\newtheorem{corollary}[theorem]{Corollary}
\newtheorem{proposition}[theorem]{Proposition}
\newtheorem{definition}[theorem]{Definition}
\newtheorem{remark}[theorem]{Remark}

\newtheorem{assumption}[theorem]{Assumption}

\crefname{lemma}{lemma}{lemmas}
\Crefname{Lemma}{Lemma}{Lemmas}
\crefname{ineq}{inequality}{inequalities}
\Crefname{Ineq}{Inequality}{Inequalities}
\creflabelformat{ineq}{#2{\upshape(#1)}#3}
\creflabelformat{Ineq}{#2{\upshape(#1)}#3}
\crefname{definition}{definition}{definitions}
\Crefname{definition}{Definition}{Definitions}
\crefname{prop}{proposition}{propositions}
\Crefname{Prop}{Proposition}{Propositions}
\crefrangeformat{equation}{(#3#1#4) --~(#5#2#6)}
\crefmultiformat{equation}{(#2#1#3)}{, (#2#1#3)}{, (#2#1#3)}{, (#2#1#3)}
\crefname{section}{Section}{Sections}
\crefname{subsubsubsection}{Section}{Sections}
\crefname{remark}{Remark}{Remarks}
\crefname{figure}{Figure}{Figures}
\crefname{table}{Table}{Tables}
\crefname{theorem}{Theorem}{Theorems}
\Crefname{theorem}{Theorem}{Theorems}
\crefname{algo}{Algorithm}{Algorithms}

\DeclareMathOperator*{\Ex}{\mathbb{E}}
\DeclareMathOperator*{\E}{\mathbb{E}}

\DeclareMathOperator*{\argmax}{\arg\!\max}
\newcommand{\Halmos}{\hfill\ensuremath{\square}}

\SetCommentSty{mycommfont}

\definecolor{lightgray}{RGB}{240,240,240}
\newmdenv[
  backgroundcolor=lightgray,
  linecolor=gray,
  linewidth=1pt,
  roundcorner=5pt,
  skipabove=10pt,
  skipbelow=10pt,
  innertopmargin=8pt,
  innerbottommargin=8pt,
  leftmargin=0pt,
  rightmargin=0pt
]{mdbox}

\input{common}

\begin{document}

\title{Calibrated Stackelberg Games:\\ Learning Optimal Commitments Against Calibrated Agents}
\author{
Nika Haghtalab\thanks{University of California, Berkeley, \texttt{nika@berkeley.edu}}\quad
Chara Podimata\thanks{Massachusetts Institute of Technology, \texttt{podimata@mit.edu}}\quad
Kunhe Yang\thanks{University of California, Berkeley, \texttt{kunheyang@berkeley.edu}}
}
\date{}

\maketitle

\allowdisplaybreaks
\begin{abstract}
We introduce \emph{Calibrated Stackelberg Games (CSGs)}, a generalization of the standard Stackelberg Games (SGs) framework.
In CSGs, a principal repeatedly interacts with an agent who (contrary to standard SGs) does not have direct access to the principal's action but instead best-responds to \emph{calibrated forecasts} about it.
This framework provides a powerful and realistic modeling tool that goes beyond assuming that agents use ad hoc and highly specified algorithms for interacting in strategic settings and instead builds on statistical foundations of forecasts and calibration.
We show that in CSGs, despite both the principal and the agent having less information than in standard SGs, the principal's optimal utility remains upper and lower bounded by the Stackelberg value of the one-shot game, in both \emph{finite} and \emph{continuous} settings.

Alongside CSGs, we develop stronger notions of calibration and corresponding algorithms that address two central challenges for calibration in game-theoretic environments.
First, achieving point-wise calibration typically incurs an error that scales exponentially with the dimension of the strategy space.
Second, the principal's convergence rate in CSGs depends critically on the adaptivity of the agent's calibration algorithm.
To address these challenges, we establish a meaningful, efficiently achievable relaxation of calibration based on conditioning on best-response regions.
This yields the first notion of calibration in games with a statistical rate that only depends on the number of agents' actions rather than the dimension of the principal's strategy space and that leads to no-swap regret for the agent.
We further develop adaptive calibration algorithms for the agents that provide fine-grained, any-time calibration guarantees against adversarial sequences, enabling the principal to achieve faster convergence in CSGs.
\end{abstract}

\input{intro-2}
\input{related}

\input{model-new}

\input{principal-algorithms-new}
\input{algos-adaptive-calibration}

\input{continuous}

\input{non-adaptive}
\input{conclusion}

\section*{Acknowledgments}
This work was supported in part by the National Science Foundation under grant CCF-2145898, by the Office of Naval Research under grant N00014-24-1-2159, a C3.AI Digital Transformation Institute grant, an Alfred P. Sloan fellowship, a Schmidt Science AI2050 fellowship, and an Amazon Research Award.

\pagebreak
\bibliographystyle{plainnat}
\bibliography{refs-new}

\newpage
\appendix
\input{app-calibration-std-def}
\input{app-swap-regr}
\input{app-learning-algos}
\input{app-adaptive-calibration}

\input{app-continuous}

\end{document}

%% file: common.tex
\newcommand{\BR}{\operatorname{BR}}

\newcommand{\br}{\mathbb{R}}

\newcommand{\bs}{\mathbb{S}}
\providecommand{\vs}{}
\renewcommand{\vs}{\mathbf{s}}
\newcommand{\vcs}{\mathbf{S}}
\newcommand{\calH}{\mathcal{H}}
\newcommand{\calI}{\mathcal{I}}
\newcommand{\vp}{\mathbf{p}}
\newcommand{\Vst}{V^{\star}}
\newcommand{\Vstr}{V^{\star}}
\newcommand{\yst}{y^\star}
\newcommand{\calF}{\mathcal{F}}

\newcommand{\Secref}[1]{\hyperref[#1]{Section \ref*{#1}}}
\newcommand{\Appref}[1]{\hyperref[#1]{Appendix \ref*{#1}}}

\newcommand{\calerr}{\mathrm{CalErr}}
\newcommand{\numP}{m}
\newcommand{\numA}{k}
\newcommand{\actionP}{\calA_P}
\newcommand{\actionA}{\calA_A}
\newcommand{\unif}{\operatorname{Unif}}
\newcommand{\bbR}{\mathbb{R}}
\newcommand{\Umax}{U_{\max}}
\newcommand\numberthis{\addtocounter{equation}{1}\tag{\theequation}}

\newcommand{\xhdr}[1]{\vspace{0mm} \noindent{\bf #1}}

\newcommand{\cF}{\mathcal{F}}
\newcommand{\cH}{\mathcal{H}}
\newcommand{\cE}{\mathcal{E}}
\newcommand{\calS}{\mathcal{S}}

\newcommand{\nrepeat}{l}
\newcommand{\ntrials}{\Phi}
\newcommand{\polytope}{P}
\newcommand{\epscal}{\varepsilon_{\text{cal}}}

\newcommand{\diam}[1]{ \texttt{diam}\!\left({#1}\right)}
\newcommand{\epochlen}{M}

\newcommand{\squishlist}{
 \begin{list}{$\bullet$}
  { \setlength{\itemsep}{0pt}
     \setlength{\parsep}{3pt}
     \setlength{\topsep}{3pt}
     \setlength{\partopsep}{0pt}
     \setlength{\leftmargin}{1.5em}
     \setlength{\labelwidth}{1em}
     \setlength{\labelsep}{0.5em} } }

\newcommand{\squishlisttwo}{
 \begin{list}{$\bullet$}
  { \setlength{\itemsep}{0pt}
    \setlength{\parsep}{0pt}
    \setlength{	opsep}{0pt}
    \setlength{\partopsep}{0pt}
    \setlength{\leftmargin}{2em}
    \setlength{\labelwidth}{1.5em}
    \setlength{\labelsep}{0.5em} } }

\newcommand{\squishend}{
  \end{list}  }

\newcommand{\reg}{\mathrm{Reg}}
\newcommand{\swapreg}{\mathrm{SwapReg}}

\newcommand{\x}{\mathbf{x}}

\newcommand{\ball}{\mathbf{B}}
\newcommand{\Pone}{{\bf (P1)}}
\newcommand{\Ptwo}{{\bf (P2)}}
\newcommand{\Cone}{{\bf (C1)}}
\newcommand{\Ctwo}{{\bf (C2)}}

\newcommand{\calA}{\mathcal{A}}
\newcommand{\cG}{\mathcal{G}}
\providecommand{\h}{}
\renewcommand{\h}{\mathbf{h}}

\newcommand{\bp}{\mathbf{p}}

\newcommand{\hst}{\mathbf{h}^{\star}}

\newcommand{\eps}{\varepsilon}
\newcommand{\adanormalhedge}{\textsc{AdaNormalHedge}}
\newcommand{\lazygdwog}{\textsc{LazyGDwoG}}
\newcommand{\approxmem}{\textsc{ApproxMem}}

\newcommand{\indicator}[1]{\mathbf{1}\{#1\}}

\newcommand{\calG}{\mathcal{G}}

\newcommand{\stepa}[1]{\overset{\rm (a)}{#1}}
\newcommand{\stepb}[1]{\overset{\rm (b)}{#1}}
\newcommand{\stepc}[1]{\overset{\rm (c)}{#1}}
\newcommand{\stepd}[1]{\overset{\rm (d)}{#1}}

\newcommand{\Ninit}{N_{\text{init}}}
\newcommand{\vol}{\text{Volume}}

\newcommand{\epsopt}{\varepsilon_{\text{opt}}}
\newcommand{\epsrobust}{\varepsilon_{\text{robust}}}
\usepackage[mathscr]{euscript}
\newcommand{\Alg}{\mathscr{A}}

%% file: intro-2.tex
\section{Introduction}\label{sec:intro}

Stackelberg games (SGs) are a canonical model for strategic principal-agent interactions, considering a principal (or ``leader'') that commits to a strategy $\h$ and an agent (or ``follower'') who observes this strategy and best responds by taking action $\BR(\h)$. These games are inspired by real-world applications such as economic policy design (where a tax policymaker establishes
rules for triggering audits before taxes are filed) and defense (where a principal allocates security resources to high-risk targets before vulnerabilities are exploited) and many more, see e.g.,~\citep{balcan2015commitment, GreenSGs,Tambe12,hardt2016strategic, dong2018strategic, chen2020learning,conitzer2006computing,xu2021dual}. By anticipating the agent's best-response, a principal who knows the agent's payoff function can calculate the \emph{optimal Stackelberg strategy} guaranteeing her the optimal utility $\Vst$, which is called the Stackelberg value. 
In recent years, \emph{repeated} SGs have gained popularity in addressing settings where the agent's payoff function is \emph{unknown} to the principal, but instead needs to be learned from repeated interactions. In this setting, the principal, who can only observe the agents' actions, aims to deploy a sequence of strategies $\h_1, \dots, \h_T$ over $T$ rounds whose average payoff is at least as good as $\Vst$, i.e., the value of her optimal strategy had she known the agent's payoffs in advance.

Repeated SGs are often studied under relatively strict assumptions on the agent's knowledge and algorithmic behavior. Examples include requiring the agent to best respond per round using $y_t = \BR(\h_t)$~\citep{balcan2015commitment,dong2018strategic}, necessitating the agent to precisely know the principal's strategy at all times (e.g., the attacker must anticipate the exact probabilistic allocation of the defender's security resources), or employing one of many online optimization algorithms whose every detail (down to the learning step size) can significantly impact the principal's utility~\citep{zrnic2021leads}.

In this paper, instead of working with restrictive or often unrealistic assumptions on the agent's knowledge and behavior, we build on foundational decision-theoretic concepts, such as \emph{forecasts} and \emph{calibration}~\citep{dawid1982well,foster1997calibrated,foster1998asymptotic}. In practice, while agents may not observe the principal's true strategies $\h_t$, they can form \emph{calibrated forecasts} --- a notion of consistency in beliefs about $\h_t$ --- to which they then best respond.
{From an uncertainty quantification point of view, calibrated forecasts provide a principled way of dealing with uncertainty of the environment so that agents can still perform well in its presence.}
Indeed, such a decision-theoretic perspective on game dynamics led to seminal results on converging to correlated and Nash equilibria in simultaneous multi-player games~\citep{foster1997calibrated,kakade2008deterministic}. Our work brings the perspective of calibrated forecasts to principal-agent games.

\subsection{Calibrated Stackelberg Games}

In this paper, we introduce \emph{Calibrated Stackelberg Games (CSGs)}---a class of games that strictly generalizes standard Stackelberg games. In a CSG, a principal and an agent repeatedly interact over $T$ time steps. As in standard repeated SGs, in each round the principal plays a strategy $\h_t$ from her strategy set.

Where CSGs depart from the standard model is in how the agent may respond. The agent, \emph{without observing $\h_t$}, forms a \emph{prediction} $\bp_t$ about the principal’s strategy. The agent then best responds to $\bp_t$ by playing $y_t \in \BR(\bp_t)$. Both players observe some information about each other’s deployed strategies and receive their corresponding utilities. These games are called ``calibrated'' because the sequence of predictions $\bp_{1:T}$ made by the agent is required to be \emph{calibrated}\footnote{See Def.~\ref{def:calibration-adaptive-CSG} for details.} with respect to the strategies actually played by the principal $\h_{1:T}$. 

Importantly, \emph{CSGs} directly generalize the standard model of repeated SGs, {capturing settings in which the agent has uncertainty about the principal’s chosen action or strategy}.

\subsubsection{Characterizing Principal's Optimal Utility.}~
In standard Stackelberg games, $\Vst$ characterizes the optimal utility of the principal.  
As a natural analog, we ask:

\xhdr{Q1. \emph{What characterizes the principal's optimal utility in CSGs?}}

Unlike in repeated Stackelberg games---where the agent’s behavior is fully specified as a best response---in CSGs the only requirement is that the agent’s predictions be \emph{calibrated}. Calibration is a common property satisfied by many forecasting procedures, rather than a specification of any particular algorithm the agent must follow. Thus, it is not even clear a priori whether calibration constrains the agent’s behavior sufficiently for the principal’s utility to be meaningfully analyzable or her optimal utility to be characterized.

Despite this lack of structural assumptions on the agent, we show that the principal’s optimal utility in CSGs converges exactly to $\Vst$\,---\,nothing more, nothing less\,---\,in games with either \emph{finite} (\Cref{sec:learning-algos}) or \emph{continuous} (\Cref{sec:continuous}) action spaces. Our answer to \textbf{Q1} thus establishes that the principal can meaningfully converge to $\Vstr$, the value she could have achieved in the one-shot Stackelberg game if she had known the agent’s utility function.

\begin{theorem}[Informal version of \Cref{thm:upper-bound,thm:ETC-regr}]
\label{thm:informal}
Assume the agent is calibrated. Then, for \emph{any} algorithm the principal uses to select $\h_{1:T}$,
\[
\lim_{T \to \infty} \frac{1}{T}\sum_{t=1}^T U_P(\h_t, y_t) \;\le\; \Vst.
\]
Furthermore, there exists a learning algorithm for the principal for choosing strategies $\h_{1:T}$ such that
\[
\lim_{T \to \infty} \frac{1}{T}\sum_{t=1}^T U_P(\h_t, y_t) \ge \Vst.
\]
\end{theorem}

Note that $\Vst$ is a benchmark that gives both players more power: the principal knows the agent's utility and the agent observes the principal's strategy. 
We find it somewhat surprising then that the optimal achievable principal utility in CSGs, in which both players work with significantly less knowledge, converges to $\Vst$ exactly.

Our definition and results also immediately apply to two important subclasses of Stackelberg games: Stackelberg Security Games~\citep{Tambe12,balcan2015commitment,haghtalab2022learning} and strategic classification~\citep{dong2018strategic}. Consequently, we obtain the first learning guarantees against calibrated agents in both of these settings.

\subsubsection{Designing Efficient and Adaptive Calibration Algorithms.}~
Producing calibrated predictions is itself an interesting and highly non-trivial problem, especially in high-dimensional settings such as games. In particular, applying the classical approach of \citet{foster1998asymptotic} with a naive discretization leads to calibration error that grows exponentially with the dimension of the strategy space. This inefficiency not only impacts the calibration rate but also affects the convergence rate to $\Vst$. 
A second challenge concerns the \emph{adaptivity} of the agent’s calibration algorithm: the principal’s convergence rate in CSGs depends critically on how quickly the agent’s forecasting procedure can adjust to a changing environment, providing the type of calibration guarantees that are valid in any window of time.

\xhdr{Q2. \emph{Are there adaptive forecasting algorithms for the agent that achieve calibration at a polynomial dependence on the size of the games?}}

We answer {\bf Q2} by introducing a general approach for obtaining a fine-grained, anytime notion of calibration with polynomial dependence on the number of strategies. This approach is of independent interest. We then specialize it to the setting of calibration in games.

To address the efficiency challenge, we provide a meaningful and efficiently achievable relaxation of calibration in CSGs by conditioning on \emph{best-response regions} rather than individual points. 
This relaxation avoids the exponential dependence on the dimensionality of the principal's strategy space and yields the \emph{first} efficiently achievable relaxation of calibration in games with a statistical rate that depends only on the number of the agent's actions\footnote{{This also yields the first efficiently achievable calibration notion that guarantees the agent incurs no \emph{swap regret} \citep{blum2007external,hart2000simple} when making decisions by best responding to the calibrated forecasts. 
Since the preliminary version of this work, \citet{noarovhigh} have extended this perspective to handle arbitrary conditioning events beyond the best-response regions of two-player games.}}.

To address the adaptivity challenge, we introduce a new notion of calibration termed \emph{adaptive calibration}. 
This notion, inspired by \emph{adaptive regret bounds} in online learning,
requires that predictions $\bp_{s:t}$ be calibrated against the realized strategies $\h_{s:t}$ in any interval $[s, t]$, with rates that depend on $(t - s)$ rather than the overall horizon $T$. 
Adaptive calibration is a strengthening of the classical notion of calibration, which applies only to the full horizon $[1, T]$. 

Moreover, {as we show in \Cref{sec:non-adaptive},} adaptive calibration bounds (with or without the aforementioned relaxation meant to promote efficiency in the size of the game)
allow the principal to achieve a faster convergence rate in $T$ in repeated interactions.

We give a principled approach for attaining adaptive calibrated forecasting. 

\begin{theorem}[Informal version of \Cref{thm:exist-adaptive-calibration}]
There exists a parameter-free forecasting algorithm that achieves adaptive calibrated forecasting on any interval $[s, t]$ with rate 
$\widetilde{O}\left(\sqrt{1/(t-s)}\right)$.\footnote{See \Cref{thm:exist-adaptive-calibration} for the dependence of this rate on the dimensionality of the problem.}
\end{theorem}

Our technique builds on insights from the multi-calibration and multi-objective literature~\citep{haghtalab2023unifying}, which connect calibration to no-regret and best-response dynamics, as well as from the sleeping-experts framework that has long been a staple of the online learning literature~\citep{blum1997empirical,freund1997using}. 
In particular, the perspective of \citet{haghtalab2023unifying} allows us to frame the design of \emph{adaptive calibration} methods as the design of a pair of no-regret and best-response algorithms in a problem with objectives that measure the predictions bias in every appropriate level-set of prediction and time sub-interval.
This, in turn, enables us to draw inspiration from adaptive regret bounds more broadly, such as the seminal AdaNormalHedge algorithm of \citet{luo2015achieving}, which ensures that a no-regret algorithm’s performance is valid on every interval $[s, t]$. 
The multi-objective framework of \citet{haghtalab2023unifying} then lifts these adaptive no-regret bounds to obtain adaptive calibrated forecasting guarantees in our setting.

Together, these results allow us to achieve calibration in games that are not only computationally tractable but also provide stronger adaptive guarantees.

%% file: related.tex
\subsection{Related work}

\paragraph{Repeated Stackelberg games.} 
Learning optimal Stackelberg strategies has been studied in the offline~\citep{conitzer2006computing} and the online setting, where only instantaneous best-responses are observable (i.e., no access to a best-response oracle). Key applications include Stackelberg Security Games (e.g., ~\citep{blum2014learning,balcan2015commitment,peng2019learning,xu2016playing}) and strategic classification (e.g.,~\cite{dong2018strategic,chen2020learning,ahmadi2021strategic,ahmadi2023fundamental}). There is a line of work that treats repeated Stackelberg games as meta-game where both players choose game-playing algorithms as their strategies, and studies optimal strategies for infinite~\citep{zuo2015optimal} or finite~\citep{arunachaleswaran2022efficient} horizons. Other works consider learning in the presence of non-myopic agents that best respond by maximizing discounted utilities~\citep{amin2013learning,haghtalab2022learning,abernethy2019learning}.
The main distinction to our work is that in our setting, the agents have only calibrated forecasts regarding the principal's strategies rather than full knowledge of them, which has been the central assumption in learning in Stackelberg games thus far.
{
Notably, our work is technically closest to that of \citet{haghtalab2022learning}, and we build on their result on optimizing the principal's utility using approximate best-response queries. Conceptually, however, our work takes a diverging perspective of assuming the agent has only minimal knowledge of the principal's strategies and use forecasting algorithms to deal with their uncertainty, whereas \citet{haghtalab2022learning} assumes the agent has full knowledge of the principal's long-term learning algorithm and their primary goal is to enable the principal to learn in the presence of non-myopic agent behavior.
}

\paragraph{Stackelberg games beyond best responses.}

Recent works have studied variants of repeated Stackelberg games with different agent strategic behaviors beyond best responding.
A prominent line of research focuses on agents employing various forms of no-regret learning, and studies how the principal’s cumulative utility compares to the one-shot Stackelberg value.
Prior works have studied agents using mean-based learning algorithms~\citep{braverman2018selling}, gradient descent~\citep{fiez2019convergence,fiez2020implicit}, no-external regret algorithms~\citep{braverman2018selling,deng2019strategizing,zrnic2021leads}, no-internal (swap)
regret algorithms~\citep{deng2019strategizing,mansour2022strategizing}, and no-counterfactual internal regret algorithms~\citep{camara2020mechanisms}.
{Our upper bound on the principal’s utility against a calibrated agent is closely related to the results of \citet{deng2019strategizing,mansour2022strategizing} showing that the principal's average reward cannot exceed the Stackelberg value when agents have no swap regret. While our result uses the calibration property directly, we also elaborate on its relationship to swap regret in \cref{app:swap-regret}.
}
Another research direction assumes agents approximately best respond due to uncertainty in the principal's strategy~\citep{blum2014lazy,an2012security,muthukumar2019robust} or their own~\citep{letchford2009learning,kiekintveld2013security,kroer2018robust} and study \emph{robust Stackelberg equilibria}~\citep{pita2010robust,gan2025robust} to optimize the principal's strategy against approximate best responding agents. {Most work in this literature assumes that the principal has either exact or noisy access to the agent’s utility function, with the exception of~\citet{zrnic2021leads}, who consider the special case of strategic classification in which the principal observes only the agent’s responses.}
The core differences of the aforementioned works to our framework are that (1) we work in an online learning setting where the principal does not have initial knowledge about the agent's utility function and has to learn from their behaviors; (2) we do not assume a specific agent algorithm but focus on properties of agent beliefs that are shared by many algorithms.

\paragraph{Calibration and learning dynamics in games.} The study of calibration, introduced by~\citet{dawid1982well}, dates back to seminal work by \citet{foster1998asymptotic,hart2022calibrated} that showed the existence of asymptotic online calibration algorithms against any adversarial sequence of events.
Applying calibration to game dynamics, \citet{foster1997calibrated} introduced the concept of \emph{calibrated learning}, which refers to a player best responding to calibrated forecasts of others' actions. They demonstrated that the game dynamics of all players performing calibrated learning converge to the set of correlated equilibria. This is complemented by the results of \cite{kakade2008deterministic,foster2018smooth,foster2021forecast} who showed that \emph{smooth and continuous} variants of calibrated learning dynamics converge to Nash equilibrium.
Our work differs from the above works by studying game dynamics that converge to the Stackelberg equilibrium, where only the follower (agent) performs calibrated learning, while the principal can observe the agents' responses.

To the best of our knowledge, our work is the first to characterize relaxed notions of calibrations for assisting an agent with no-(swap) regret decision-making in repeated games, 
with a statistical rate that only depends on the number of the agent's actions rather than the dimension of the principal's strategy space.

Since the publication of the preliminary version of this paper, this has become an important research direction in the literature, and there has been rapid progress both in studying calibration for repeated decision-making and in understanding calibrated learning as a behavioral assumption for multi-agent interactions.
On the calibration front, \citet{noarov2023high,noarovhigh} develop efficient algorithms for \emph{online multicalibration}, ensuring that any downstream agent who best responds to the forecasts achieves vanishing (swap) regret, which generalize our best response correspondence to more general class of events.  
A recent line of work focuses on developing efficiently achievable calibration measures that simultaneously guarantees no regret decision-making for all downstream agents who best respond to the calibrated predictions according to their own utility functions, where \citet{kleinberg2023u,luo2024optimal} focuses on minimizing external regret, and
\citet{roth2024forecasting,hu2024predict} focuses on minimizing swap regret. Recent advances also provide efficient algorithms for high-dimensional calibration~\cite{peng2025high} or swap regret minimization~\cite{peng2024fast,dagan2023external}, where they achieve improved dependency on the dimension (number of actions) but with a weaker dependency on the time horizon.

There has also been recent developments in understanding repeated interactions in games under different behavioral models together with the benchmark for measuring \emph{performance} in different learning dynamics. The concurrent work of \citet{brown2023learning} studies agents that are $\Phi$-regret-minimizing; in their work, they provide an algorithm that learns the Stackelberg equilibria in unknown games against no-adaptive-$\Phi$-regret agents, but they did not specify the convergence rates of said algorithm. Note that $\Phi$-no-regret is a more general behavioral model than calibrated best-response, and hence, we expect our convergence rates to be better {and more specialized to calibrated agents}. \citet{collina2024efficient} considers an extended setting where the utilities in each round are state-dependent and the state is unknown to both the principal and the agent, and they show that calibrated forecasts of the state can be used to remove the common prior assumption against no counterfactual-internal-regret agents. Another line of work investigates the principal's benchmarks and learning algorithms against agents that are long-term rational in choosing their adaptive strategies. \citet{ananthakrishnan2024knowledge} investigate whether the principal can overcome information asymmetry by analyzing equilibria of the meta-game between game-playing algorithms. \citet{arunachaleswaran2024pareto,arunachaleswaran2025learning} study the principal's learning algorithms against unknown agents that are Pareto-optimal or utility-maximizing in a Bayesian setting.

\paragraph{Adaptivity and sleeping experts.}
The notion of adaptive calibration that we introduce in \Cref{sec:model} is related to notion of adaptivity of regret bounds in online learning~\citep{luo2015achieving,daniely2015strongly,jun2017improved}.
Our design of adaptively calibrated forecasting algorithms builds on the \emph{multi-objective learning} perspective of online (multi-)calibration~\citep{lee2022online,haghtalab2023unifying} and the powerful tool of \emph{sleeping experts}~\citep{blum1997empirical,freund1997using,luo2015achieving} which have proven useful in various applications such as fairness~\citep{blum2020advancing}. {These works are methodologically related to our Section~\ref{sec:algos-adaptive-calibration} algorithms', although we \emph{do} need to \emph{carefully} adapt them in order to achieve our notion of adaptive calibration.}

%% file: model-new.tex
\section{Model \& Preliminaries}\label[section]{sec:model}

We begin this section with some basic definitions about forecasts, calibration, and games, and then introduce the class of games that we study: \emph{Calibrated Stackelberg Games} (CSGs).
\subsection{Forecasts, Calibration, and Games}

\newcommand{\actionset}{A}
\newcommand{\forecastset}{C}
\paragraph{Adaptively Calibrated Forecasts.} 

We use $\actionset$ to denote the space of {outcomes} and $\forecastset\supseteq \actionset$ to denote the space of forecasts. 
A (stochastic) forecasting procedure $\sigma$ {is an online procedure that takes any adversarial  sequence of outcomes $\h_t\in\actionset$ for $t\in [T]$, and on round $t$ outputs (possibly at random) forecast $\bp_t\in \forecastset$, based on outcomes and forecasts $\h_\tau,\bp_\tau$, for $\tau\in [t-1]$.}

To define calibrated forecasts, let us first introduce the notion of \emph{binning functions}.

\begin{definition}[Binning~\citep{foster2021forecast}]
    \label[definition]{def:binning}
    We call a set $\Pi = \{w_i\}_{i \in [n]}$ a \emph{binning function}, if each $w_i: \forecastset \to [0,1]$ maps forecasts to real values in $[0,1]$, and for all $\bp\in\forecastset$ we have $\sum_{i\in[n]}w_i(\bp)=1$.
\end{definition}

{With the above binning functions, we define the adaptive calibration error with respect to $\Pi$ as follows. At a high level, conditioned on any bin, the calibration error measures the difference between the expected forecasts that fall in that bin and the corresponding expected outcome.
\begin{definition}[$\Pi$-Adaptive Calibration Error]
    \label[definition]{def:calibration-error}
    For any time interval $[s,t]$, let $\bp_{s:t}$ be the sequence of forecasts and $\h_{s:t}$ be the sequence of outcomes.
    For a given binning $\Pi=\{w_i\}_{i\in [n]}$ with size $n$, and $\ \forall i\in [n]$, define the $\Pi$-adaptive calibration error as
    \begin{align}
        &\calerr_i\left(\h_{s:t},\bp_{s:t}\right)\triangleq
            \frac{n_{[s,t]}(i)}{t-s}\cdot\left\|\bar{\bp}_{[s,t]}(i)-\bar{\h}_{[s,t]}(i)\right\|_{\infty},
            \label{eq:def-calibration-error}
    \end{align}
where during interval $[s,t]$, $n_{[s,t]}(i)\triangleq\sum_{\tau=s}^t w_i(\bp_\tau)$ is the effective number of times that the forecast belongs to bin $i$ (i.e., bin $i$ is activated), $\bar{\bp}_{[s,t]}(i) \triangleq\sum_{\tau=s}^t \frac{w_i(\bp_\tau)}{n_{[s,t]}(i)}\cdot \bp_\tau$ is the expected forecast that activates bin $i$, $\bar{\h}_{[s,t]}(i)\triangleq\sum_{\tau=s}^t \frac{w_i(\bp_\tau)}{n_{[s,t]}(i)}\cdot {\h_{\tau}}$ is the expected outcomes corresponding to bin $i$.
\end{definition}
We say that a forecasting procedure is adaptively calibrated if it achieves vanishing calibration error on any adversarial sequence of outcomes and any sub-interval of time steps. 
\begin{definition}[$(\eps, \Pi)$-Adaptively Calibrated Forecasts]
    \label{def:calibration-adaptive}
    
    A forecasting procedure $\sigma$ is $\eps$-\emph{adaptively calibrated} to binning $\Pi=\{w_i\}_{i\in [n]}$ with rate $r_{\delta}(\cdot)\in o(1)$, if for all adversarial sequences of actions $\h_1,\cdots,\h_T$, where $\h_t\in \actionset$, $\sigma$ outputs forecasts $\bp_t\in \forecastset$ for $t\in [T]$  such that with probability at least $1-\delta,$ we have that 
$\forall s,t$ such that $1\le s<t\le T$, and $\ \forall i\in [n]$:
    \begin{align*}        \calerr_i\left(\h_{s:t},\bp_{s:t}\right)\le r_\delta(t-s)+\eps.
    \end{align*}
\end{definition}
}

We remark that without adaptivity (i.e., instead of requiring calibration error to be small for any sub-interval $[s,t]$, only requiring it to be small on the entire interval $[1,T]$),~\Cref{def:calibration-error} is weaker than the standard definition of calibration (e.g., \citep{foster1998asymptotic}, listed for completeness in Appendix~\ref{app:calibration-std}) in two ways: (1) standard calibration takes each prediction $\bp\in\forecastset$ as an independent bin, thus having infinitely many binning functions: $w_{\bp}(\cdot)=\delta_{\bp}(\cdot)$. Instead, we only require calibration with respect to the predefined binning $\Pi$ which only contains a finite number of binning functions; (2) standard calibration cares about the summation over calibration error across bins, but we only consider the maximum error.

\paragraph{Stackelberg Games.}
A \emph{Stackelberg game} is defined as the tuple $(\calA_P, \calA_A, U_P, U_A)$, where $\calA_P$ and $\calA_A$ are the principal and the agent action spaces respectively, and $U_P: \calA_P \times \calA_A \to \bbR_+$ and $U_A: \calA_P \times \calA_A \to \bbR_+$ are the principal and the agent utility functions respectively. For ease of exposition, we work with \emph{finite} Stackelberg games (i.e., $|\calA_P| = \numP$ and $|\calA_A| = \numA$) and generalize our results to continuous games in Section~\ref{sec:continuous}. When the principal plays action $x \in \calA_P$ and the agent plays action $y \in \calA_A$, then the principal and the agent receive utilities $U_P(x,y)$ and $U_A(x,y)$ respectively. We also define the principal's \emph{strategy space} as the simplex over actions: $\calH_P = \Delta(\calA_P)$. For a strategy $\h \in \calH_P$, we oftentimes abuse notation slightly and write $U_P(\h, y) := \E_{x \sim \h} [U_P(x,y)]$. 

\emph{Repeated Stackelberg games} capture the \emph{repeated} interaction between a principal and an agent over $T$ rounds. What distinguishes Stackelberg games from other types of games is the inter-temporal relationship between the principal's action/strategy and the agent's response; specifically, the principal first commits to a strategy $\h_t \in \calH_P$ and the agent subsequently \emph{best-responds} to it with $y_t \in \calA_A$. Let $\vp_t \in \calF_P = \calH_P$ be the agent's \emph{belief} regarding the principal's strategy at round $t$. In standard Stackelberg games: $\vp_t = \h_t$, i.e., the agent has full knowledge of the principal's strategy. In this paper, we consider games where the agent does not in general know $\h_t$ when playing, but they only best-respond according to their belief $\vp_t$.
The agent's \emph{best-response} to \emph{belief} $\vp_t$ according to her underlying utility function $U_A$ is action $y_t \in \calA_A$ such that

\begin{equation}
    y_t \in \BR(\vp_t) \quad \text{where} \quad \BR(\vp_t) = \argmax_{y \in \calA_A} \E_{x \sim \vp_t} [U_A(x,y)].
    \label{eq:br-cal}
\end{equation}
We often overload notation and write $U_A(\vp,y) := \E_{x \sim \vp}[U_A(x,y)]$. Note that from Equation~\eqref{eq:br-cal}, the best-responses to $\vp_t$ form \emph{a set}. If this set is not a singleton, we use either a \emph{deterministic} or a \emph{randomized tie-breaking} rule. For the \emph{deterministic} tie-breaking rule, the agent breaks ties according to a predefined preference rule $\succ$ over the set of actions $\calA_A$. For the \emph{randomized} tie-breaking rule, the agent chooses $y_t$ by sampling from the set $\BR(\vp_t)$ uniformly at random, i.e., $y_t \sim \unif(\BR(\vp_t))$.

The \emph{Stackelberg value} of the game is the principal's optimal utility when the agent best responds: 
\[
V^\star=\max_{\h^\star\in\cH_P}\max_{y^\star\in\BR(\h^\star)} U_P(\h^\star,y^\star).
\]
In the above definition $\h^\star$ is referred to as the \emph{principal's optimal strategy.}

For an agent's action $y \in \calA_A$, we define the corresponding \emph{best-response polytope} $\polytope_y$ as the set of all of the agent's beliefs that induce $y$ as the agent's best-response, i.e., $\polytope_y = \{ \vp \in \calF_P: y \in \BR(\vp)\}$. We make the following standard assumption, which intuitively means that there are sufficiently many strategies that induce $y^\star$ as the agent's best-response.
\begin{assumption}[Regularity]
    \label{assumption:regularity}
    The principal's optimal strategy $\h^\star\in\Delta(\actionP)$ and the agent's optimal action $y^\star\in\BR(\h^\star)$ satisfy a \emph{regularity condition}: \emph{$\polytope_{y^\star}$ contains an $\ell_2$ ball of radius $\eta>0$}. 
\end{assumption}

\begin{algorithm}[htbp]
    \caption{Interaction Protocol for Calibrated Stackelberg Games (CSGs)}
    \label[algo]{alg:protocol}
    \DontPrintSemicolon
    \LinesNumbered
    \SetAlgoNoLine
         The principal plays strategy $\h_t \in \cH_P$.\;
         The agent \emph{without observing $\h_t$} forms a \emph{calibrated prediction} $\bp_t$ (Def.~\ref{def:calibration-adaptive-CSG}) about $\h_t$.\;
         The agent best responds to $\bp_t$ by playing $y_t \in \BR(\bp_t)$ (including tie-breaking).\;
         The principal observes $y_t$ and experiences utility $U_P(\h_t,y_t)$.\;
         The agent observes $\h_t$, or an action sampled from $\h_t$. 
        
\end{algorithm}

\subsection{Calibrated Stackelberg Games}
In CSGs (see Algorithm~\ref{alg:protocol} for the principal-agent interaction protocol,

the agent forms $(\eps, \Pi)$-adaptively calibrated forecasts as their beliefs $\vp_t$ regarding $\h_t$.
{Note that if the agent observes action $x_t\sim \h_t$ instead of the mixed strategy $\h_t$, then they can still calibrate to the sequence of $\h_t$ with an additional (vanishing) error term that comes from concentration inequalities.}

We first define binning functions that are especially appropriate for forecasts in games.
In CSGs, we define $\Pi$ based on whether $i$ is a best-response to the input calibrated forecast, i.e., $\forall \vp \in \calF_P$: 
    \begin{align*}
    w_i(\vp) &= \indicator{i \in \BR(\vp), i \succ j, \forall j \neq i} &\tag{for the \emph{deterministic} tie-breaking} \\
    {w_i}(\vp) &= \frac{\indicator{i \in \BR(\vp)}}{|\BR(\vp)|} &\tag{for the \emph{randomized} tie-breaking}
    \end{align*}
{Note that both binning functions meet the conditions of \Cref{def:binning}.

}
Applying \Cref{def:calibration-adaptive} for calibrated agent forecasts in CSGs we have the following:

\begin{definition}[$\eps$-Adaptively Calibrated Agent for CSGs]
    \label[definition]{def:calibration-adaptive-CSG}
The agent is called $\eps$-adaptively calibrated with rate $r_{\delta}(\cdot)\in o(1)$, if 
for any sequence of principal strategies $\h_1,\cdots,\h_T\in\cH_P$ the agent takes a sequence of actions $y_1, \dots, y_T$ that satisfy the following requirements: 
1) there is a sequence of forecasts $\bp_t\in \cF_P$ for $t\in [T]$, such that $y_t \in \BR(\bp_t)$, and 
2) forecasts $\bp_1, \dots, \bp_T$ are $\eps$-calibrated for binning $\Pi$ with rate $r_\delta(\cdot)$ with respect to the principal's strategies $\h_1,\cdots,\h_T$.

\end{definition}
We next review the fundamental constructs from \Cref{eq:def-calibration-error} and their intuitive meaning in this setting. $n_{[s,t]}(i)\triangleq\sum_{\tau \in [s,t]} w_i(\bp_\tau)$ is now the expected number of times that the forecast has induced action $i$ from the agent as their best response during interval $[s,t]$, $\bar{\bp}_{[s,t]}(i) \triangleq\sum_{\tau \in [s,t]} {w_i(\bp_\tau)}\cdot \bp_\tau/{n_{[s,t]}(i)}$ 
is the expected forecast that induces action $i$ from the agent as their best response during interval $[s,t]$, and $\bar{\h}_{[s,t]}(i) \triangleq\sum_{\tau \in [s,t]} {w_i(\bp_\tau)}\cdot {\h_{\tau}}/{n_{[s,t]}(i)}$ is the expected principal strategy that induces action $i$ from the agent as their best response during interval $[s,t]$. The requirement for an agent to be calibrated is quite mild, as the forecasts are binned only according to the best-response they induce.

\subsection{Background: Learning Stackelberg strategies from best responses}\label{subsec:BR-oracle}
In this section, we provide some background on computing the Stackelberg strategy from best response queries, which is a key ingredient in developing the principal's learning algorithm in \Cref{sec:learning-algos}. 

Recall that in a Stackelberg game, the principal's strategy space $\cH_P$ (the simplex over the principal's actions $\calA_P$) can be partitioned into \emph{best response polytopes} $\polytope_y$ for each agent action $y \in \calA_A$, where $\polytope_y$ contains all the principal's strategies that induce $y$ as the agent's best response. 
Inside each best response polytope $\polytope_y$, the principal's utility function $U_P(\h,\BR(\h))$ is linear in $\h$: $U_P(\h, y) = \sum_{x\in\calA_P} \h_x\cdot U_P(x, y)$. Therefore, when the boundaries of the polytopes $\{\polytope_y\}_{y\in\calA_A}$ are known, the Stackelberg strategy can be computed via a \emph{multiple LP approach}~\citep{conitzer2006computing}: for each action $y \in \calA_A$, solve $\h_y \in \argmax_{\h \in \polytope_y} U_P(\h, y)$ to find the optimal strategy $\h_y$ inside polytope $\polytope_y$, and then select the best among these candidates, i.e., choose $y^\star\in\argmax_{y\in\calA_A} U_P(\h_y, y)$ and commit to $\h^\star = \h_{y^\star}$.

When the agent's best response polytopes $\{\polytope_y\}_{y\in\calA_A}$ are unknown to the principal a priori, the principal can still compute a \emph{near-optimal Stackelberg strategy} through repeated interactions with a best response oracle --- an oracle that returns $\BR(\h)$ on each queried strategy $\h\in\cH_P$ and thus provides a membership oracle for the best response polytopes $\polytope_y$. Prior work~\cite{letchford2009learning,blum2014learning,blum2015learning} has developed efficient algorithms for this setting by leveraging techniques from convex optimization with membership queries~\citep{kalai2006simulated,lee2018efficient}. In these approaches, each subproblem of optimizing $U_P(\h, y)$ over $\polytope_y$ is solved using a membership oracle for $\polytope_y$, yielding an $\eps$-approximate Stackelberg strategy for the principal with probability $1-\delta$, using a number of best response queries that is polynomial in $\numP,\numA$ and logarithmic in $1/\eps$ and $1/\delta$. See \citet{blum2015learning} for more details on these results.

In CSGs, however, the principal lacks access to an exact best response oracle. In each round $t$, she only observes the agent's action $y_t$, which is a best response to the \emph{forecasted strategy} $\bp_t$ of $\h_t$. Note that $\bp_t$ might differ from $\h_t$ on a per-round basis, but the agent's calibration guarantees ensure that $\bp_t$ is consistent with $\h_t$ on a time-average basis. We address this challenge in \Cref{sec:learning-algos} by extending the results on learning Stackelberg strategies from best response queries to deal with \emph{approximate best response queries}~\citep{haghtalab2022learning}, and by robustifying the learned strategy against the agent's calibration error. We provide additional notations for deriving robust strategies in the next subsection.

\subsection{Robust strategies}
\begin{wrapfigure}{r}{0.3\textwidth}
    \centering
    \includegraphics[trim={35cm 14cm 15cm 11cm},clip,width=0.3\textwidth]{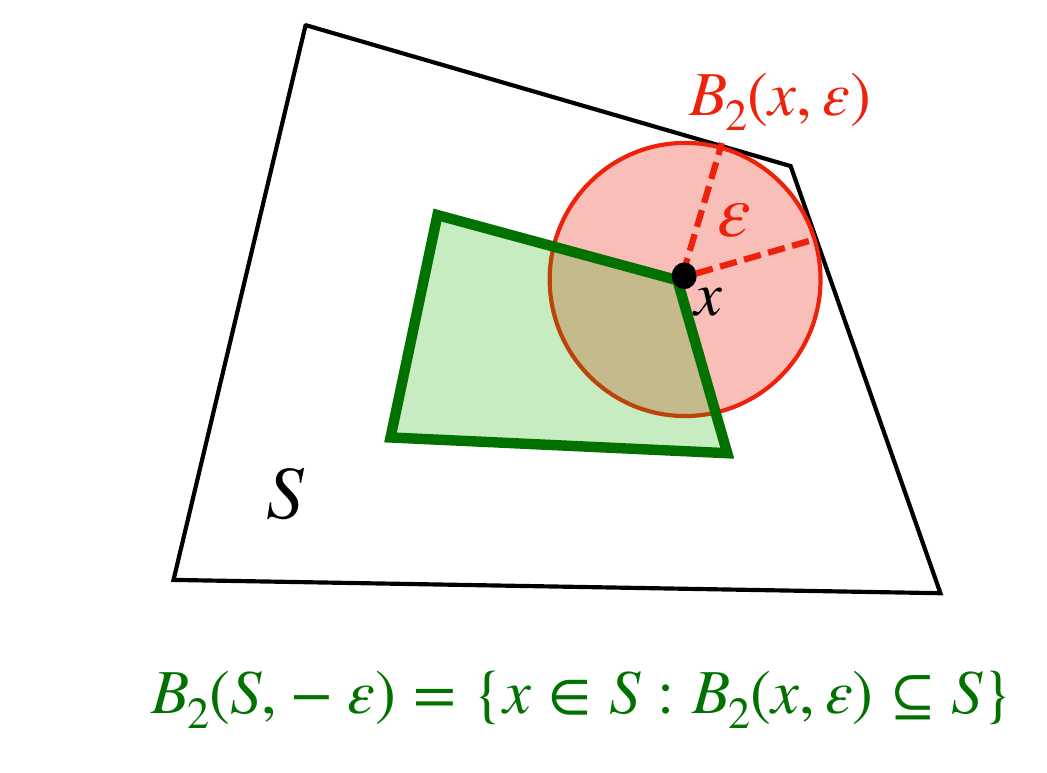}
    \caption{\footnotesize Conservative region $B_2(S,-\eps)$.}
    \label{fig:ball}
\end{wrapfigure}
In this section, we introduce additional notations that will be used for deriving robust strategies for the principal against the agent's calibration error.
Let $B_2(x,\eps)$ denote the ball of radius $\eps$ around $x$, i.e., $B_2(x,\eps) \triangleq \{x': \|x - x'\|_2 \leq \eps\}$. 
For a convex set $S \in \bbR^n$, we use $B_2(S, -\eps)$ to denote the set of all points in $S$ that are ``safely'' inside $S$, i.e., all the points $\x$ whose neighborhood of radius $\eps$ is contained in $S$: $B_2(S,-\eps)\triangleq \left\{x\in S:\ B_2(x,\eps)\subseteq S\right\}$. We call this last set, the $\eps$-\emph{conservative} of $S$. When $S$ is a best response polytope $\polytope_y$, we also use the shorthand $\polytope_y^{-\eps} \triangleq B_2(\polytope_y, -\eps)$ to denote the $\eps$-\emph{conservative} of $\polytope_y$.
See \Cref{fig:ball} for a pictorial illustration.

%% file: principal-algorithms-new.tex
\section{Principal's learning algorithms}\label{sec:learning-algos}
    In this section, we study the relationship between the principal's Stackelberg value $\Vst$ and the best utility the principal can obtain from learning to play a sequence of strategies $\{\h_t\}_{t \in [T]}$ against calibrated agents, i.e., $\frac 1T \sum_{t \in [T]} U_P(\h_t, y_t)$. 
    The relationship between $\Vst$ and $\frac 1T \sum_{t \in [T]} U_P(\h_t, y_t)$ is \emph{not} a priori clear. 
    In the case of calibrated forecasts, the agents do not know the exact $\h_t$ when they choose their response. Instead, they base their decisions on the history of the principal's strategies. A principal then may be able to create historical patterns that lead the agents to worse actions, thus obtaining better utility himself.
    Indeed, several works have shown how historical patterns can afford the principal much better utility than $\Vst$ when the agents are no regret~\citep{braverman2018selling,deng2019strategizing}.
    Surprisingly, we show that this is not the case when the agents are calibrated; $\sum_{t \in [T]} U_P(\h_t, y_t)$ is upper bounded by $T \Vst$ and a term that is sublinear in $T$ and depends on the calibration parameters. A similar upper bound on the principal's utility was proved for no-swap-regret agents~\citep{deng2019strategizing}. While we prove the theorem directly for calibration, an alternative proof in App.~\ref{app:swap-regret} shows that calibration implies no-swap regret.

    \begin{theorem}
        \label{thm:upper-bound}
        Assume that the agent is $(\eps,\Pi)$-adaptively calibrated with rate $r_\delta(\cdot)$ and let $\Umax = \max_{\h \in \cH_P} \max_{y \in \calA_A} U_P(\h, y)$. Then, for any sequence $\{\h_t\}_{t \in [T]}$ for the principal's strategies in a repeated CSG, with probability at least $1 - 2 \delta$, the principal's utility is upper bounded as: 
        \[
        \frac{1}{T}\sum_{t \in [T]} U_P(\h_t, y_t) \leq \Vst + \alpha(U_{\max}, \numP, \numA, T, r_\delta, \delta, \eps),
        \]
        where $\alpha(U_{\max}, \numP, \numA, T, r_\delta, \delta, \eps) = U_{\max} \numP \numA  ( r_\delta(T) + \eps )$ when the agent uses \emph{deterministic} tie-breaking and $\alpha(U_{\max}, \numP, \numA, T, r_\delta, \delta, \eps) = U_{\max}  \numP \numA  \left( r_\delta(T) + \eps + \sqrt{\frac{\log (1/\delta)}{T}}\right)$ when the agent uses \emph{randomized} tie-breaking.
    \end{theorem}

    \proof{Proof of \Cref{thm:upper-bound}.}
    To simplify notation, we use $n_T(i) := n_{[0,T]}(i)$. When the agent follows deterministic tie-breaking, we have:
        \begin{align*}
            \sum_{t=1}^T U_P(\h_t,y_t)=&\sum_{i\in\actionA}\sum_{t=1}^T w_i(\bp_t) U_P(\h_t,i)\\
            =&\sum_{i\in\actionA} n_T(i)  U_P(\bar{\h}_T(i),i) &\tag{linearity of $U_P$ in the principal's strategy}\\
            =&\sum_{i\in\actionA} n_T(i)  \left(U_P \left(\bar{\bp}_T(i),i\right) + \left\langle U_P(\cdot,i),\bar{\h}_T(i) - \bar{\bp}_T(i)\right\rangle\right) &\tag{$\pm \sum_i U_P(\bar{\bp}_T(i),i)$}\\
            \stepa{\le} &\sum_{i\in\actionA} n_T(i) V^\star + \sum_{i\in\actionA} n_T(i)\left\|\bar{\bp}_T(i)-\bar{\h}_T(i)\right\|_{\infty}\cdot\|U_P(\cdot,i)\|_1\\
            \stepb{\le} &V^\star T+U_{\max}\numP \sum_{i\in\actionA} T\cdot\calerr_i(\h_{1:T},\bp_{1:T})\\
            \le &V^\star T+U_{\max}\cdot \numA \cdot \numP \cdot T \cdot (r_\delta(T)+\eps)
        \end{align*}
        where (a) is because $U_P(\bar{\bp}_T(i),i) \leq \Vst$ (since $i \in \BR(\bar{\bp}_T)$), the Hölder's inequality, and the fact that $\|\cdot \|_2 \leq \| \cdot \|_1$, and (b) is because of the definition of $\Umax$ and \Cref{def:calibration-adaptive}.
        The proof for the randomized tie-breaking setting has an extra term from Azuma-Hoeffding's inequality, similar to the proof of \Cref{lemma:no-swap-regret}.
    \Halmos

On the other hand, it may seem that because the agent's behavior is less specified when she uses calibrated forecasts (as opposed to full knowledge of the principal's $\h_t$), the principal may only be able to extract much less
utility compared to $\Vst$. Again, we show that this is not the case and that there exist algorithms for the principal such that the sequence of strategies $\{\h_t\}_{t \in [T]}$ is asymptotically approaching $\Vst$.

\begin{restatable}{theorem}{thmETCregr}\label[theorem]{thm:ETC-regr}
    There exists an {efficient} algorithm ({Algorithm~\ref{algo:explore-commit}}) for the principal in CSGs that achieves average utility: $\lim_{T \to \infty }\frac{1}{T}\sum_{t \in [T]} U_P(\h_t,y_t)\ge V^\star.$ 
    
    For finite $T$, if the agent's calibration rate is $r_\delta(t)=O(t^{-\beta})$, {and the best response polytope associated with the agent's optimal action contains a ball with radius $\eta$ (see Assumption~\ref{assumption:regularity})},
    then the algorithm can guarantee the principal's expected average utility to be at least: 
    \[
        \Ex\left[\frac{1}{T}\sum_{t \in [T]} U_P(\h_t,y_t)\right] \geq \Vst - O\left(
        \numA^{\frac{1}{14}}\,m\,\eta^{-\frac{13}{14}}\,T^{-\frac{1}{14\beta}}
        \;+\;\frac{k^{\frac{\beta}{14}}\,m^{\beta+\frac{1}{2}} }{\vol(\eta/2)}\eta^{-\frac{13\beta}{14}}\,T^{-\frac{1}{14}}\log T\right),
    \]
\end{restatable}

{where $\vol(\eta/2)$ stands for the volume of the ball of radius $\eta/2$.}

\begin{algorithm}[t!]
    \caption{Explore-Then-Commit}
    \label[algo]{algo:explore-commit}
    \DontPrintSemicolon
    \LinesNumbered
    \SetAlgoNoLine
    \KwIn{Target precision $\epsopt$, target robustness $\epsrobust$, time horizon $T$.}
    \textsc{Explore}: Find $\epsopt$-optimal, $\epsrobust$-robust strategy $\hat{\h} \in \cH_P$ using Algorithm~\ref{algo:principal-main}.\;
    \textsc{Commit}:
        Repeatedly play $\hat{\h}$ for the rest of the rounds.\;
    \end{algorithm}

\medskip

\paragraph{Proof and Technical Challenges Overview.} Algorithm~\ref{algo:explore-commit} is an explore-then-commit algorithm; it first estimates an \emph{appropriate} strategy for the principal $\hat{\h}$ (\textsc{Explore}), and then repeatedly plays it until the end (\textsc{Commit}). In the remainder of the section, we {walk the reader through} the main steps and challenges of the proof. Let $T_1, T_2$ denote the set of rounds that belong in the \textsc{Explore} and \textsc{Commit} phase respectively. 

To elaborate on the objectives of the \textsc{Explore} phase, let us first consider a setting with zero calibration error, where the agent's forecasting algorithm is perfectly and adaptively calibrated, leading to $y_t=\BR(\h_t)$ at every round. In this case, the task for the \textsc{Explore} phase simplifies to identifying a near-optimal strategy $\hat{\h}$ through best response oracles (see Section~\ref{subsec:BR-oracle}) that satisfies $U_P(\hat{\h},\BR(\hat{\h}))\ge \Vst-\epsopt$ for a predetermined $\epsopt$. We formalize this property in \Pone. Given that the agent is perfectly calibrated, in the \textsc{Commit} phase, the agent always plays $\hat{y}=\BR(\hat{\h})$, leading to an upper bound of $\epsopt|T_2|$ on the Stackelberg regret. Hence, for a perfectly calibrated agent, Algorithm~\ref{algo:explore-commit}'s regret is bounded by $\Vst|T_1|+\epsopt|T_2|$.
\begin{mdbox}
\begin{itemize}
    \item[\Pone] $U_P(\hat{\h}, \hat{y}) \geq \Vst - \epsopt$ for $\hat{y} \in \BR(\hat{\h})$, i.e., $(\hat{\h},\hat{y})$ is an approximate Stackelberg equilibrium.
\end{itemize}
\end{mdbox}
\vspace{-10pt}

Moving away from the idealized setting, we must account for possible discrepancies between $y_t$ and $\BR(\h_t)$ due to calibration error. This introduces two challenges: (i) An increased sample complexity $|T_1|$ in the \textsc{Explore} phase, given the necessity to learn a near-optimal strategy from \emph{noisy} agent responses; (ii) Potential agent deviations from the action $\hat{y}=\BR(\hat{\h})$ in the \textsc{Commit} phase due to miscalibrations in the agent's belief about the principal's action. 

To address the first challenge, we employ Algorithm~\ref{algo:principal-main}, which constructs an \emph{approximate} best response oracle by repeatedly interacting with a calibrated agent. {Algorithm~\ref{algo:principal-main} is rather involved, and we defer its more detailed discussion to after we have finished the general proof outline. For now, the reader should know that Algorithm~\ref{algo:principal-main} returns a tuple $(\hat{\h}, \hat{y})$ satisfying {\bf (P1)} in $|T_1|$ number of rounds.} {The precise number of rounds $|T_1|$ (which is determined later in \Cref{app:proof-etcreg}) is tuned to balance the exploration and exploitation trade-off.}

For the second challenge, we require our learned policy $\hat{\h}$ to be robust against inaccurate forecasts. This is reflected in condition \Ptwo, which necessitates the ball of radius $\epsrobust$ around $\hat{\h}$ to be fully contained in the polytope $\polytope_{\hat{y}}$. The critical insight from \Ptwo\ is: for any forecast $\bp_t$ that results in a best response $y_t=\BR(\bp_t)\neq\hat{y}$, there must be a minimum distance of $\epsrobust$ separating $\bp_t$ from $\hat{\h}$.

We formalize $\Ptwo$ below.
For appropriately tuned parameters $\epsopt$ and $\epsrobust$, the pair $(\hat{\h},\hat{y})$ returned by Algorithm~\ref{algo:principal-main} satisfies properties \Pone\ and \Ptwo:
\begin{mdbox}
    \begin{itemize}
\item[\Ptwo] $\hat{\h} \in B_2 (P_{\hat{y}}, -\epsrobust)$, i.e., $\hat{\h}$ lies \emph{robustly} within the best-response polytope for $\hat{y}$.
\end{itemize}
\end{mdbox}
\vspace{-10pt}

Given \Ptwo, the regret of the \textsc{Commit} phase can be decomposed to when $y_t = \hat{y}$, and when $y_t \neq \hat{y}$:
\begin{equation}\label{eq:bef-p1-p2}
    \sum_{t \in T_2} (\Vst - U_P(\hat{\h}, y_t) )=\sum_{t\in T_2:y_t=\hat{y}}(\Vst - U_P(\hat{\h}, y_t) )+\sum_{t\in T_2:y_t\neq\hat{y}}(\Vst - U_P(\hat{\h}, y_t) )
\end{equation}

When $y_t=\hat{y}$, the regret comes from the approximate best-response oracle guarantees; \Pone\ guarantees that $\Vst - U_P(\hat{\h}, \hat{y} ) \leq \epsopt$, so the first term is at most $\epsopt\cdot|T_2|$.

When $y_t\neq \hat{y}$, the regret is driven by the fact that the agent best-responds to calibrated forecasts of the principal's actions. Let $A = \actionA\setminus\{ \hat{y}\}$. 
For $i\in A$, the definition of binning function $w_i$ guarantees that the probability of playing action $i$ on forecast $\bp_t$ is exactly $w_i(\bp_t)$.
Based on this observation, the second term of \Cref{eq:bef-p1-p2} can be further bounded as
\begin{align*}
    \sum_{t \in T_2} \sum_{i \in A} w_i(\vp_t) (\Vst - U_P(\hat{\h}, i)) \leq \sum_{i \in A} \sum_{t \in T_2} w_i(\vp_t) \Vst=\sum_{i\in A}n_{T_2}(i) \Vst.\numberthis{\label{eq:bef-apply-p2}}
\end{align*}

Using \Cref{def:calibration-error} of the calibration error and properties of the $\ell_2$ and $\ell_{\infty}$ norms, we can further express $n_{T_2}(i)$ as follows
\begin{align*}
    n_{T_2}(i)= \frac{\calerr_i(\h_{T_2},\bp_{T_2})\cdot|T_2|}{\|\bar{\vp}_{T_2}(i) - \hat{\h}\|_\infty}\le
    \frac{\sqrt{m} \cdot \calerr_i(\h_{T_2},\bp_{T_2}) \cdot |T_2|}{\|\bar{\vp}_{T_2}(i) - \hat{\h}\|_2} \overset{\Ptwo}{\leq} \frac{\sqrt{m} \cdot r_{\delta}(|T_2|) \cdot |T_2|}{\epsrobust},
\end{align*}
where for the second inequality is: since $\hat{\h}$ lies in the $\epsrobust$-conservative of $\polytope_{\hat{y}}$, and $\bar{\vp}_{T_2}(i)$ belongs to a different and non-intersecting polytope $\polytope_i$, we know that $\|\bar{\vp}_{T_2}(i) - \hat{\h}\|_2\ge\epsrobust$. See \Cref{fig:robust} for a geometric interpretation. Finally, {bounding the sample complexity for the \textsc{Explore} phase (\Cref{lemma:principal-main-complexity}) and appropriately selecting $\epsopt,\epsrobust$ gives the result.}

\begin{figure}[htbp]
  \begin{center}
    \includegraphics[trim={20cm 9cm 26cm 15cm},clip,width=0.4\textwidth]{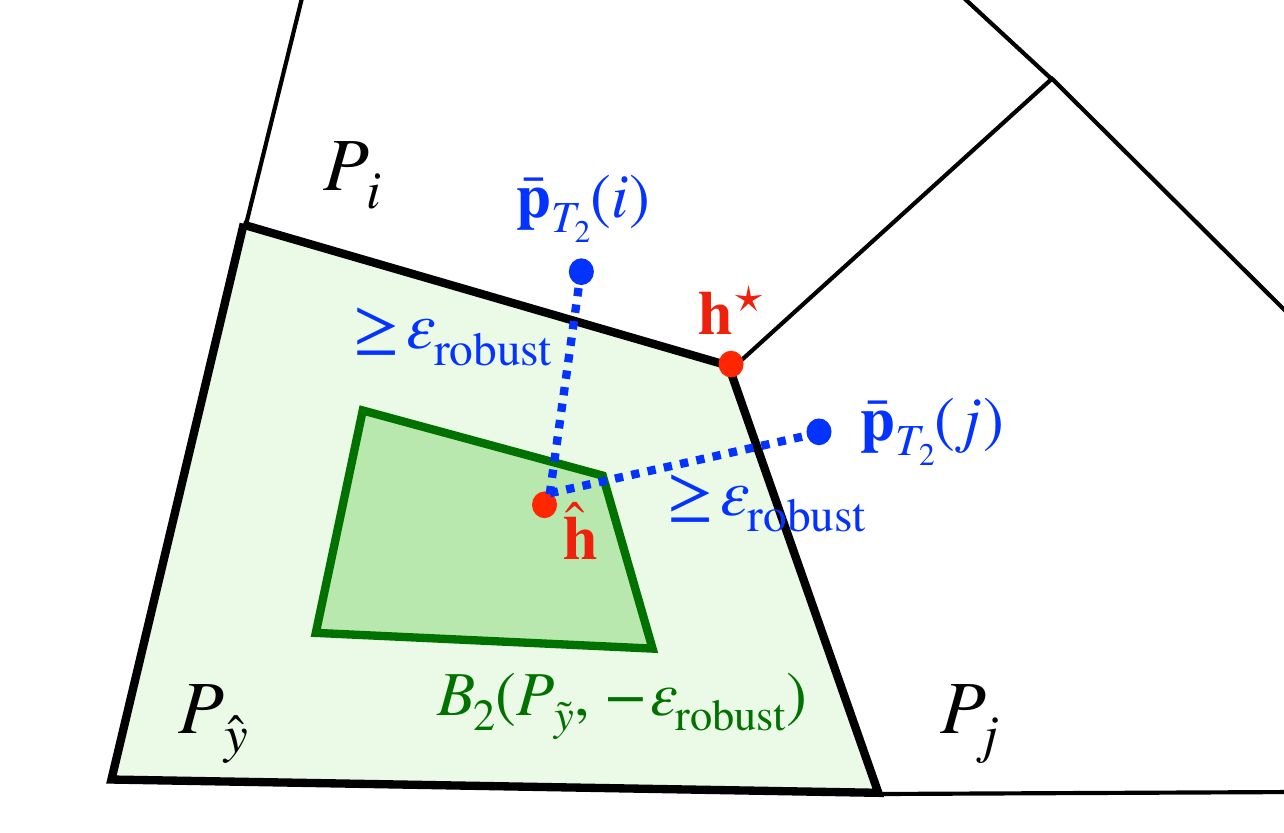}
    \end{center}
    \caption{Relationship between the robust policy $\hat{\h}$ and the average predictions $\bar{\mathbf{p}}_{T_2}(i)$: Given that $\hat{\h}$ is in the conservative region $B_2(P_{\hat{y}},-\epsrobust)$, any average prediction $\bar{\mathbf{p}}_{T_2}(i)$ that triggers action $i\neq\hat{y}$ during the \textsc{Commit} phase must fall outside of $P_{\hat{y}}$ and thus have a distance of at least $\epsrobust$ from $\hat{\h}$.}
  \label{fig:robust}
\end{figure}
 
\paragraph{Algorithm 3 Overview.} The proof sketch above hinges on being able to identify a strategy for the principal $\hat{\h}$ 

with properties \Pone \hspace{3pt}$\&$ \Ptwo. We outline below how this is done through Algorithm~\ref{algo:principal-main}, and defer the detailed analysis of its guarantees to \Cref{app:learning-algos}.

\begin{algorithm}[H]
    \caption{Principal's Learning Algorithm for the Optimal Commitment}
    \label[algo]{algo:principal-main}
    \DontPrintSemicolon
    \LinesNumbered
    \SetAlgoNoLine
    \SetAlgoNoEnd
    \KwIn{
        Target precision $\epsopt$, target robustness $\epsrobust$,

        regularity parameter $\eta$,
        Approximate membership oracle $\approxmem$ (Algorithm~\ref{algo:conservative-approx-MEM}).
    
    }
    \KwOut{A pair $(\hat{\h},\hat{y})$ such that $U_P(\hat{\h},\hat{y})\ge\Vst-\epsopt$ and $\hat{\h}\in B_2(\polytope_{\hat{y}},-\epsrobust)$.}
    \tcc{Initialization phase (\Cref{lemma:initialization-set}): building initialization set $\calI$ with pairs $(\h_i,y_i)$, where each $\h_i$ is well-centered in $\polytope_{y_i}$.}
    Originally, create initialization set $\calI \gets  \emptyset$.\;
    $\Ninit\gets 10\frac{\log T}{\vol(\eta/2)}$.\;
    \For{$i\in[\Ninit]$}{
        Sample a strategy $\h_i \in \cH_P$ uniformly at random.\;
        $y_i\gets\approxmem(\eps_1= \frac{\eta}{4}, \eps_2=\frac{\eta}{4\sqrt{\numP}} , \eps_3=\frac{1}{10T^2})$ with input $\h_i$\;
        
        \textbf{if} $y_i\neq\perp$, \textbf{then}  $\calI \gets \calI \cup (y_i,\h_i)$.
        \tcp{$\h_i$ lies robustly inside $\polytope_{y_i}$ with high probability.}
    }

    \tcc{Optimization phase: for each pair $(y_i,\h_i)\in\calI$, optimize the principal's utility $U_P(\cdot, y_i)$, using $\h_i$ as initial point}
    ${\calS}\gets\emptyset$.
        \tcp{certified robust solutions}
    \For{$(y_i,\h_i) \in \calI$}{ 
      Solve the following program using the \textsc{MembershipOpt} algorithm~\citep{haghtalab2022learning}, with 
      initial point $\h_i$
      and oracle $\approxmem(\eps_1= C\sqrt{\numP}\left(\frac{\eta\cdot\epsopt}{m}\right)^{13},\ \eps_2={\epsrobust}{\sqrt{\numP}},\ \eps_3=\frac{1}{10T^2})$:
      \[\max_{\h}U_P(\h,y_i),\quad\text{subject to }\h\in B_2(\polytope_{y_i},-\epsrobust).
      \]

        $\Tilde{\h}_{i}\gets$ solution returned by \textsc{MembershipOpt}.\;
        \tcc{Robustness check: ensure that the returned solution $\Tilde{\h}_{i}$ is robustly inside the best-response polytope for $y_i$.}
        Query $\approxmem(\eps_1= \frac{\epsrobust\sqrt{\numP}}{2} ,\ \eps_2=\frac{\epsrobust}{2},\ \eps_3=\frac{1}{10T^2} )$ with input $(\Tilde{\h}_{i},y_i)$.
        \textbf{if} \textsc{True}, \textbf{then}  add $(\Tilde{\h}_{i},y_i)$ to ${\calS}$; \textbf{else} discard the pair $(\Tilde{\h}_{i},y_i)$.
        }

        $(\hat{\h},\hat{y})\gets\arg\max_{(\Tilde{\h}_{i},y_i)\in{\calS}} U_P(\Tilde{\h}_{i},{y_i})$\;
        \textsc{Return} $(\hat{\h},\hat{y})$\;
        
    \end{algorithm}

Algorithm~\ref{algo:principal-main} is structured into three main steps: initialization, optimization, and robustness check.

\begin{enumerate}
    \item \textbf{Initialization}: The algorithm begins by constructing an \emph{initialization set} $\calI$. This set is composed of pairs $(\h, y)$, where each $\h$ is a strategy sampled uniformly at random from the principal's strategy space $\cH_P$. For each sampled strategy $\h$, the algorithm only uses it as initial point if it lies \emph{robustly} within the best response polytope for action $y$, ideally when $\h\in B_2(\polytope_y,-\frac{\eta}{2})$ for a given regularity parameter $\eta>0$. While the algorithm cannot confirm this with certainty, we are able to construct \emph{approximate best response oracles} (Algorithm~\ref{algo:conservative-approx-MEM}) to certify this property up to some approximation factor.
    If $\h$ turns out to be well-centered, the pair $(\h, y)$ is added to the initialization set $\calI$. Otherwise, the pair is discarded.

    \item \textbf{Optimization}: With the initialization set $\calI$ prepared, the algorithm proceeds to the optimization phase. For each pair $(\h_i, y_i)$ in $\calI$, the algorithm aims to maximize the principal's utility $U_P(\h, y_i)$ subject to $\h$ being robust, i.e., $\h\in B_2(\polytope_{y_i},-\epsrobust)$. Although the principal's utility function $U_P$ is known, the challenge lies in not knowing the boundaries of the conservative polytope $\polytope_{y_i}$, which depends on the agent's utility function and can only be learned through repeated interactions. To address this, we again use $\approxmem$ 
    as a membership oracle to the conservative polytope $\polytope_{y_i}$, and
    employ the \textsc{MembershipOpt} algorithm from \citep{haghtalab2022learning}, which is guaranteed to find near-optimal solutions that satisfy the robust constraint. However, this guarantee might no longer hold if the initial point is not enough well-centered, which we are not always able to detect. Therefore, we introduce the following robustness check step to verify that the solution is indeed robust.

    \item \textbf{Robustness Check}: After solving the optimization problems for all pairs in $\calI$, the algorithm performs a robustness check. This involves verifying that the solutions obtained lie robustly inside the corresponding best response polytope. If the robustness check passes, the pair $(\h_i, y_i)$ is added to the optimization set $\calS$. Otherwise, the pair is discarded. The final strategy is chosen as the one giving highest principal utility within the set $\calS$.
\end{enumerate}

As the above outline illustrates, the algorithm crucially relies on being able to construct approximate membership oracle to best response polytopes, which specifies whether a given strategy has large enough margin to the polytope boundaries.

We use $\approxmem$ (Algorithm~\ref{algo:conservative-approx-MEM} in \Cref{sec:approx-mem}) to achieve this.

Specifically, on input $\h\in\calH_P$, $\approxmem$ either asserts $\h \in B_2(\polytope_y, -\eps_2 - \frac{\eps_1}{\sqrt{\numP}})$ or $\h \notin B_2(\polytope_y, -\eps_2 - \eps_1)$ with probability at least $1 - \eps_3$. To do this, it samples $\ntrials$ points in proximity to $\h$ and plays each one repeatedly for $\nrepeat$ rounds, while registering the best-response action observed for each one of these. If the most frequent best-response for all $\h_\phi$ is $y$, then we can conclude with good probability that $\h$ lies robustly inside $\polytope_y$. See \Cref{lem:apx-memb-oracle} for more details.

Through the three main components, Algorithm~\ref{algo:principal-main} returns a tuple that satisfies both $\Pone$ and $\Ptwo$. We formally state the guarantee of the correctness and sample complexity of Algorithm~\ref{algo:principal-main} below, and defer its proof to \Cref{app:proof-principal-main-complexity}.

\begin{restatable}[Guarantee of Algorithm~\ref{algo:principal-main}]{lemma}{lemmaPrincipalMainComplexity}
    \label[lemma]{lemma:principal-main-complexity}
    Suppose $\epsrobust\le\frac{C}{\sqrt{\numP}}\cdot\left(\frac{\eta\cdot\epsopt}{\numP}\right)^{13} $ for the universal constant $C$ provided in \citep{haghtalab2022learning}, and the agent has calibration rate $r_\delta(\cdot)$. 
    Algorithm~\ref{algo:principal-main} terminates within 
    \begin{align*}
        O\left(\frac{\sqrt{\numP}}{\vol(\eta/2)} r_\delta^{-1}\left(\min\left\{\frac{\eta}{4\numA{\numP}},\frac{\epsrobust}{2\numA\sqrt{\numP}}\right\}\right)\log^2 T\right)
    \end{align*}
    rounds, and with probability at least $1-T^{-1}$, 
    returns a pair $(\hat{\h},\hat{y})$ that satisfies the following guarantees:
\begin{enumerate}
    \item[\Pone] $U_P(\hat{\h}, \hat{y}) \geq \Vst - \eps_{\text{opt}}$, i.e., $\hat{\h}$ is a near-optimal strategy.
    \item[\Ptwo] $\hat{\h} \in B_2 (P_{\hat{y}}, -\eps_{\text{robust}})$ , i.e., $\hat{\h}$ lies \emph{robustly} within the best-response polytope for $\hat{y}$.
\end{enumerate}
\end{restatable}

%% file: algos-adaptive-calibration.tex
\section{Forecasting Algorithm for Adaptive Calibration}
\label{sec:algos-adaptive-calibration}

In this section, we examine whether there exist natural forecasting procedures that satisfy our \Cref{def:calibration-adaptive} about adaptively calibrated forecasts. We answer this question positively.

\begin{restatable}{theorem}{existadaptivecalibration}
\label{thm:exist-adaptive-calibration}
    For all $\eps>0$ and all binnings $\Pi=\{w_i:\br^{\numP}\to[0,1],i\in[\numA]\}$, there exists a {parameter-free} forecasting procedure that is $(\eps, \Pi)$-adaptively calibrated with rate $r_\delta(t)={O}\left(\sqrt{\log(\numA \numP t)/t}+\sqrt{\log(T/\delta)/t}\right)$. Moreover, when $\Pi$ is a continuous binning (i.e., each $w_i$ is continuous), there exists a forecasting procedure that is $(0, \Pi)$-adaptively calibrated with the same rate.
\end{restatable}

To prove the theorem, we use two main tools; the first one is a well-known algorithm of \citet{luo2015achieving} (\adanormalhedge) applied for online learning in the \emph{sleeping experts} problem (see Appendix~\ref{app:adaptive-calibration} for details). Roughly speaking, the \emph{sleeping experts} is a standard online learning problem with $T$ rounds and $N$ experts, where at each round $t$ there is only a subset of the experts being ``awake'' to be considered by the learner and report their predictions. Let $I_{t,i}$ be the binary variable indicating whether expert $i$ was awake at round $t$ ($I_{t,i} = 1$) or asleep ($I_{t,i}=0$). The interaction protocol between the learner and the adversary at each round $t$ is: (i) The learner observes which experts are awake, i.e., $\{I_{t,i}\}_{i \in [N]}$. (ii) The learner selects a probability distribution $\pi_t \in \Delta([N])$ supported on the set of \emph{active} experts $A_t \triangleq \{i: I_{t,i} = 1\}$. (iii) The adversary selects a loss vector $\{\ell_{t,i}\}_{i \in [N]}$. (iv) The learner incurs expected loss $\hat{\ell}_t = \E_{i \sim \pi_t} [\ell_{t,i}]$. \adanormalhedge\ is a \emph{parameter-free} online learning algorithm that when applied on the sleeping experts problem (and with appropriate initialization) obtains regret $\reg_t(i) = O ( \sqrt{T_i \log (NT_i)})$,
where $T_i = \sum_{\tau \in [t]} I_{\tau, i}$. 

The second tool that we use is \emph{No-Regret vs. Best-Response dynamics (NRBR)}~\citep{haghtalab2023unifying}. NRBR are a form of no-regret dynamics between two players, where one of the players must also best-respond on average. Essentially, at each round $t \in [T]$, the forecasting algorithm with the calibration rate of \Cref{thm:exist-adaptive-calibration} outputs a randomized forecast $\vp_t \in \calF_P$, 
by simulating an interaction between two players described below. For the first player, we construct a \emph{sleeping experts} problem instance, where the set of experts is $\calG =\{g_{(s,i,j,\sigma)}:s\in[T],i\in\actionA,j\in \actionP,\sigma\in\{\pm1\}\}$, i.e., we create a different expert for each round, each principal-agent action pair, and each $\sigma$ (the use of which will be made clear in the next paragraph). For each $g_{(s,i,j,\sigma)}\in\cG$ and $t\in[T]$, we define the loss, sleeping/awake indicator, and instantaneous regret respectively as:
\begin{align*}
    \ell_{t,g_{(s,i,j,\sigma)}}\triangleq&L_{g_{(s,i,j,\sigma)}}(\h_t,\bp_t)=w_i(\bp_t)\cdot\sigma\cdot\left(h_{t,j}-p_{t,j}\right);\numberthis{\label{eq:sleep-expert-loss}}\\
    I_{t,g_{(s,i,j,\sigma)}}\triangleq&\indicator{t\ge s}\numberthis{\label{eq:sleep-expert-indicator}};\\
    r_{t,g}\triangleq&I_{t,g}\cdot (\ell_{t,g}-\hat{\ell}_t).
\end{align*}
where by $h_{t,j}, p_{t,j}$ we denote the $j$-th coordinate of $\h_t$ and $\vp_t$ respectively. 

Note that we defined the losses for our newly constructed sleeping experts' instance as above to make sure that there is a direct correspondence with the calibration error. 
Intuitively, the loss at round $t$ for expert $g_{(s,i,j,\sigma)}$ is the difference between the agent's prediction $\vp_t$ and the true event $\h_t$ at the $j$-th coordinate, multiplied by the sign $\sigma$ and the weight $w_i(\bp_t)$ --- which, specialized to the setting of games, is the best response action that prediction $\bp_t$ activates. 
Similar ideas for calibration (albeit not for the notion of adaptivity {we consider}) have been used in \citep{lee2022online,haghtalab2023unifying}. We describe next the player interaction in NRBR.
For adaptive calibration, we introduce the sleeping/awake indicator $I_{t,g_{(s,i,j,\sigma)}}$ to guarantee that each expert $g_{(s,i,j,\sigma)}$ becomes awake at round $s$ and stays awake until the end of the interaction. 

We now describe the player interaction in NRBR.

\paragraph{Player 1.} Runs $\adanormalhedge$ on expert set $\calG$ with a pre-specified prior $\pi_0(g_{(s,i,j,\sigma)})\propto\frac{1}{s^2}$  over $\cG$ and feedback specified in \cref{eq:sleep-expert-loss,eq:sleep-expert-indicator}. At each round $t$, Player 1 computes distribution $\pi_t\in \Delta(A_t(\cG))$, where $A_t(\cG)$ denotes the set of active experts $g_{(s,i,j,\sigma)}\in\cG$ with $I_{t,g_{(s,i,j,\sigma)}}=1$ (i.e., experts that satisfies $s\le t$).

\paragraph{Player 2.} Best responds to $\pi_t$ by selecting $Q_t\in \Delta(\calF_P)$ that satisfies:
    \begin{align}\label{eq:minmax}
        \max_{\h_t\in\calH_P}\Ex_{g\sim \pi_t \atop\vp_t \sim Q_t}\left[\ell_{t,g}\right]
        =\max_{\h_t\in\calH_P}\Ex_{g\sim \pi_t \atop\vp_t \sim Q_t}\left[L_{g}(\h_t,\bp_t)\right]
        \le \eps.
    \end{align}
After simulating the game above, the algorithm outputs forecast $\vp_t \sim Q_t$. The existence of such a distribution $Q_t$ is justified by the min-max theorem (\citep[Fact 4.1]{haghtalab2023unifying} or \citep[Theorem 5]{foster2021forecast}). In the Appendix, we also give an explicit formula for $Q_t$ in the special case of $\numP\!=\!2$. {When $\Pi_0$ is continuous, player 2 can select a deterministic $\bp_t$ that achieves Equation~\eqref{eq:minmax} with $\eps = 0$. This stronger property is justified by the outgoing fixed-point theorem~\citep[Theorem 4]{foster2021forecast}}. 
Note that this algorithm inherits its parameter-free property directly from $\adanormalhedge$.
We are now ready to provide the proof for \Cref{thm:exist-adaptive-calibration}.

\proof{Proof of \Cref{thm:exist-adaptive-calibration}.}

Fix an instance of the NRBR game outlined above.

Running $\adanormalhedge$ on the instance with $\calG$ that we specified above, with prior $\pi_0(g_{(s,i,j,\sigma)})\propto\frac{1}{s^2}$ \citep[Section 5.1]{luo2015achieving} guarantees that $\forall g_{(s,i,j,\sigma)}\in\calG$, the regret with respect to $g_{(s,i,j,\sigma)}$ is bounded by:
\begin{equation}\label{eq:reg-AdaNormalHedge-SE}
\reg_t\left(g_{(s,i,j,\sigma)}\right)=\sum_{\tau \in [t]} r_{t,g_{(s,i,\sigma)}}\le \widetilde{O}\left(\sqrt{(t-s)\log(\numA\numP T)}\right).
\end{equation}
where $\tilde{O}(\cdot)$ hides lower order poly-logarithmic terms. 

We proceed by translating the calibration error defined in \Cref{def:calibration-adaptive} to the sleeping experts instance that we defined above. 

We have that $\forall i\in[\numA]$ and $1\le s<t\le T$, the calibration error can be written as:
\begin{align*}
&\calerr_i(\h_{s:t},\bp_{s:t})=\frac{1}{t-s}\max_{j\in\actionP}\max_{\sigma\in\{\pm1\}}
    \sum_{\tau=s}^t I_{\tau,g_{(s,i,j,\sigma)}} \cdot \ell_{\tau,g_{(s,i,j,\sigma)}} &\tag{\Cref{def:calibration-adaptive}}\\
    & = \frac{1}{t-s}\max_{j \in \calA_P} \max_{\sigma \in \{\pm 1\}}\underbrace{\sum_{\tau=s}^t I_{\tau,g_{(s,i,j,\sigma)}} \cdot \left( \ell_{\tau,g_{(s,i,j,\sigma)}} - \E_{g \sim \pi_\tau}\left[ \ell_{\tau, g} \right] \right)}_{\reg_t\left(g_{(s,i,j,\sigma)}\right)} \\
    & \qquad \qquad \qquad + \frac{1}{t-s}\max_{j \in \calA_P} 
    \max_{\sigma \in \{\pm 1\}}\sum_{\tau=s}^t I_{\tau,g_{(s,i,j,\sigma)}} \cdot \E_{g \sim \pi_\tau}\left[ \ell_{\tau, g} \right] \\
    &= \frac{1}{t-s}\max_{j\in\actionP}\max_{\sigma\in\{\pm1\}}\reg_t\left(g_{(s,i,j,\sigma)}\right) + \frac{1}{t-s} 
    \sum_{\tau = s}^t\E_{g \sim \pi_\tau} \left[ L_{g}(\h_\tau, \vp_\tau)\right],
    \numberthis{\label{eq:CallErr-1}}
    \end{align*}
    where for the first derivation we add and subtract $\sum_\tau \hat{\ell}_\tau$ and use that because of the NRBR dynamics: $\hat{\ell}_\tau = \E_{g \sim P_\tau} \left[ \ell_{\tau, g} \right]$, and for the last derivation we have used the definition of $\ell_{\tau,g} = L_g (\h_\tau, \vp_\tau)$ (\Cref{eq:sleep-expert-loss}), as well as the fact that $I_{\tau,g_{(s,i,j,\sigma)}}=1$ when $\tau\ge s$, for all $j\in\numP$ and $\sigma\in\{\pm1\}$. 
    We have thus decomposed the calibration error into two terms: the regret of \textbf{Player 1} from running $\adanormalhedge$ on the sleeping experts instance, and the loss of \textbf{Player 2} from selecting the final forecast $\vp_\tau\sim Q_\tau$. In the following, we upper bound the second term from \textbf{Player 2} using the fact that $Q_\tau$ is obtained by best responding to $\pi_\tau$ (see \Cref{eq:minmax}). We have with probability at least $1-\delta$, for each $j\in\actionP$, $\sigma\in\{\pm1\}$, and $1\le s\le t\le T$,
    \begin{align*}
        &\ \frac{1}{t-s} 
        \sum_{\tau = s}^t\E_{g \sim \pi_\tau} \left[ L_{g}(\h_\tau, \vp_\tau)\right] \\
    & = \frac{1}{t-s}\sum_{\tau = s}^t \E_{g \sim \pi_\tau \atop\bp\sim Q_\tau } \left[ L_g (\h_\tau, \vp)\right]+ \frac{1}{t-s}\sum_{\tau = s}^t \E_{g \sim \pi_\tau} \left[ L_g(\h_\tau, \vp_\tau) - \E_{\vp \in Q_\tau} \left[ L_g(\h_\tau, \vp) \right]  \right] \\
        &\leq \frac{1}{t-s} \sum_{\tau = s}^t \max_{\h_\tau \in \Delta(\calA_P)} \E_{g \sim \pi_\tau \atop\bp\sim Q_\tau } \left[ L_g (\h_\tau, \vp)\right] + \frac{1}{t-s}\sum_{\tau = s}^t \E_{g \sim \pi_\tau} \left[ L_g(\h_\tau, \vp_\tau) - \E_{\vp \in Q_\tau} \left[ L_g(\h_\tau, \vp) \right]  \right] \\
        &\leq \eps + O\left(\sqrt{\frac{\log (T/\delta)}{t-s}}\right)
    \end{align*}
    where the first inequality is by the property of $\h_\tau$ being the best strategy for the principal, and the last one uses the fact that $\max_{\h_\tau \in \Delta(\calA_P)} \E_{g \sim \pi_\tau \atop\bp\sim Q_\tau } \left[ L_g (\h_\tau, \vp)\right] \leq \eps$ from the NRBR \Cref{eq:minmax} and a martingale concentration bound on the second term. 

Plugging the upper bound for $Q$ back to Equation~\eqref{eq:CallErr-1} and using the regret bound for AdaNormalHedge (\Cref{eq:reg-AdaNormalHedge-SE}) we get: 
\begin{align*}
\calerr_i(\h_{s:t},\bp_{s:t}) &\leq \widetilde{O} \left( \sqrt{\frac{\log ( \numA \numP T)}{t-s}} \right) + O \left( \sqrt{\frac{\log (T/\delta)}{t-s}}\right) + \eps \leq r_t(\delta) + \eps.\Halmos
\end{align*}

\endproof

%% file: continuous.tex
\section{Continuous Games}\label[section]{sec:continuous}

In this section, we generalize our results for the case of \emph{continuous} Stackelberg games. To streamline presentation, this section highlight how our results generalize to continuous games; the technical details (including all proofs) can be found in Appendix~\ref{app:continuous}.

\paragraph{Continuous Stackelberg Games.} We use again $\calA_P$ and $\calA_A$ to denote the principal and the agent action spaces, respectively. Both $\calA_A, \calA_P$ are convex, compact sets where $\actionP\subset\br^{\numP}$ and $\actionA\subset\br^{\numA}$. The utilities of the principal and the agent are given by continuous functions $U_P:\actionP\times\actionA\to\br_+$ and $U_A:\actionP\times\actionA\to\br_+$. 
In this setting, we assume that both the principal and the agent can only play deterministic strategies, i.e., $\calH_P=\actionP$.
For $x\in\actionP$, let $\BR(x)$ be the best-response function that is implicitly defined as $\nabla_2 U_A(x,\BR(x))=0$. Our continuous games satisfy Assumption~\ref{assump:continuous}: (i)-(iii) are standard assumptions used in previous works (e.g.,~\citep{fiez2019convergence}), but (iv) cannot be derived from (i) and (ii) without further assumptions on the correlation between $x,y$. Nevertheless, (iv) (and the conditions under which it holds) has been justified in settings such as strategic classification~\citep{dong2018strategic,zrnic2021leads}.
\begin{assumption}
\label{assump:continuous}
Utility functions $U_P$, $U_A$, and the domain $\actionP$ satisfy the following:
    \squishlist
        \item[(i)] For all $x\in\calA_P,y\in\actionA$, $U_P(x,y)$ is $L_1$-Lipschitz and concave in $x$, $L_2$-Lipschitz in $y$, and bounded by $W_P$ in $\ell_2$ norm.
        \item[(ii)]  The best-response function $\BR:\actionP\to\actionA$ is $L_{\BR}$-Lipschitz.
        \item[(iii)] Regularity of the feasible set $\calA_P=\calH_P=\calF_P$:
        \squishlist
            \item The diameter is bounded: $\texttt{\upshape diam}({\calF_P})=\sup_{\h,\h'\in\calF_P}\|\h-\h'\|_2\le D_P$.
            \item $B(0,r)\subseteq \actionP\subseteq B(0,R)$.
        \squishend
        \item[(iv)] The function $U_P(\h,\BR(\h))$ is concave with respect to $\h$, and has Lipschitz constant $L_{U}$. 
    \squishend
\end{assumption}

The main result of this section is to show that even in \emph{continuous} CSGs, we can approximate asymptotically $\Vst$ for the principal's utility, and that no better utility is actually achieved. 

\begin{restatable}{theorem}{maincontinuous}\label{thm:main-continuous}
    For continuous CSGs satisfying Assumption~\ref{assump:continuous}, for all $\eps_0 > 0$, there exists a finite binning $\Pi_0$ such that if the agent is $(0, \Pi_0)$ - adaptively calibrated and the principal runs an appropriately parametrized instance of $\lazygdwog$ (Algorithm~\ref{algo:gd}) then: 
    \begin{align*}
        \lim_{\Phi\to\infty\atop \epochlen\to\infty}\frac{1}{\Phi\epochlen}\sum_{\phi\in[\Phi]}\sum_{i\in[\epochlen]}U_P(\h_\phi,y_{\phi,i})\ge V^\star-\eps_0.
    \end{align*}
    Moreover, for any sequence of the principal's actions $\h_{[1:T]}$, it holds that: 
    \begin{align*}
    \lim_{T\to\infty}
    \frac{1}{T}\sum_{t\in[T]} U_P(\h_t,y_t)\le V^\star+\eps_0.
\end{align*}
\end{restatable}
Recall that we can define $(\eps,\Pi)$-adaptive calibration with $\eps = 0$ in continuous CSGs because of the continuous case in Theorem~\ref{thm:exist-adaptive-calibration}.

We outline next how $\lazygdwog$ works.
\begin{algorithm}[t]
    \caption{Lazy Gradient Descent without a Gradient (\lazygdwog)}
    \label[algo]{algo:gd}
    \DontPrintSemicolon
    \LinesNumbered
    \SetAlgoNoLine
    \SetAlgoNoEnd
    Initialize $\h_0=0$.\;
    \For{epoch $\phi\ge0$: 
    }{
    Sample $S_\phi$ uniformly at random from the unit sphere $\mathbb{S}^{\numP-1}$.\;
    Play $\h_\phi=\x_\phi+\delta_\phi S_\phi$ for $\epochlen$ rounds. \tcc*{avg feedback gets close to $\BR(\h_t)$}
    Observe agent's responses $y_{\phi,1},\cdots,y_{\phi,\epochlen}$.\;
    Update action $\x_{\phi+1}\gets \text{Proj}_{B_2(\actionP,-\delta_\phi)}({\x_\phi+\gamma_\phi \frac{\numP}{\delta_\phi}  S_\phi
    U_P(\h_\phi, \frac{1}{\epochlen}\sum_{i\in[\epochlen]}y_{\phi,i})}).$ 
    }
    \end{algorithm}
$\lazygdwog$ is a variant of the gradient descent without a gradient algorithm (\textsc{GDwoG}) of \citet{Flaxman2004OnlineCO}. The main new component of the algorithm is that it separates the time horizon into epochs and for each epoch it runs an update of the \textsc{GDwoG} algorithm. During all the rounds that comprise an epoch ($M$ in total), $\lazygdwog$ presents the same (appropriately smoothed-out) strategy to the agent and observes the $M$ different responses from the agent. The intuition behind repeating the same strategy for $M$ rounds is that the principal wants to give the opportunity to the agent to recalibrate for a better forecast, i.e., $\lim_{M \to \infty}\frac{1}{\epochlen}\Big|\{i\in[\epochlen]:\|\bp_i-\h\|\ge\eps_0\}\Big| = 0$. The remainder of the proof for \Cref{thm:main-continuous} centers on showing that when the calibrated forecasts converge to $\h_t$, the principal's utility converges to the utility they would have gotten if the agent was perfectly best responding to $\h_t$.

%% file: non-adaptive.tex
\section{Extension to Non-Adaptively Calibrated Agents}
\label{sec:non-adaptive}

In this section, we extend our results to agents whose forecasts are calibrated in a \emph{non-adaptive} sense. Recall that adaptive calibration (\Cref{def:calibration-adaptive}) requires the agent's forecasts to achieve vanishing calibration error on every subinterval $[s,t]\subseteq[1,T]$, with $\calerr_i(\h_{s:t},\bp_{s:t})$ scaling as a function of the interval length $(t-s)$. 
This adaptivity ensures that the agent's forecasts can adjust promptly to changes in the principal's strategy, and thus enables the principal to test different strategies on different sub-intervals and thereby learn the (near-)optimal strategy at a faster rate, as we did in \Cref{sec:learning-algos}.

In contrast, non-adaptive calibration only requires shrinking calibration error on prefixes of the interaction. Formally, we say that an agent is $\eps$-\emph{nonadaptively} calibrated with rate $r_\delta(\cdot)\in o(1)$ if,
with probability at least $1-\delta$, for every $i\in [\numA]$ and all $1\le t\le T$, we have
\[
\calerr_i(\h_{1:t},\bp_{1:t}) \le r_\delta(t)+\eps.
\]
While non-adaptive calibration ensures that the agent's forecasts become increasingly accurate in the long run, it provides no guarantees on subintervals that do not start at time step $1$. As a result, the principal cannot expect the agent's forecasts to adapt promptly to newly deployed strategies, making it more challenging to learn from such agents.

In \Cref{subsec:nonadaptive-reduction}, we show that the principal's exploration algorithm (Algorithm~\ref{algo:principal-main}) can be adapted to interact with nonadaptively calibrated agents via a reduction that incurs only a \emph{polynomial} overhead in the number of required interactions. In \Cref{subsec:separation-adaptive-nonadaptive}, we discuss whether such a polynomial overhead is necessary. While we do not establish a formal separation result, we present a simple illustrative instance demonstrating that nonadaptive calibration may fail to provide timely information about best responses on intermediate subintervals. This example highlights the challenges posed by nonadaptive calibration and illustrates why stronger adaptive guarantees lead to faster learning rates.

\subsection{Principal's Algorithm for Non-Adaptively Calibrated Agents}
\label{subsec:nonadaptive-reduction}

In this section, we derive a reduction that leverages our results for adaptively calibrated agents to interact with non-adaptively calibrated agents. The key idea is to design alternative time schedules for the principal's learning algorithm, which already has a blockwise structure that is well suited for such rescaling. Specifically, the principal's algorithm in Algorithm~\ref{algo:principal-main} can be viewed as a sequence of blocks: in each block, the principal plays a fixed strategy in every round of that block and updates only at block boundaries. The strategy for the next block is chosen adaptively based on the most frequent action that the agent played during each of the previous blocks. Our reduction modifies the length of each block, while keeping the principal's update rule across blocks unchanged. Now we formally state the reduction.

\begin{lemma}[Time-rescheduling for blockwise majority-based procedures]
    \label{lem:time-rescheduling}
    Consider a blockwise algorithm $\Alg$ that contains $M$ blocks $I_1, I_2, \dots, I_M$, and succeeds whenever the calibration error in each block is at most $\gamma$, i.e.,
    \begin{align}
        \label{eq:blockwise-condition}
    \max_{j\in[\numA]} \calerr_j\!\left(\h_{I_i},\bp_{I_i}\right)\le \gamma,
    \qquad
    \forall i\in[M].
    \end{align}
    Suppose the agent is \emph{nonadaptively} calibrated with nonincreasing rate $r_\delta(\cdot)$, 
    then there exists a time-rescheduled version of $\Alg$, denoted $\Alg^{\mathrm{NA}}$, that runs the same decision rule as $\Alg$ at block boundaries, but uses block lengths $L_i \triangleq r_\delta^{-1}\!\left(\frac{\gamma}{2i}\right)$ for each block $i\in[M]$, and succeeds with probability at least $1-\delta$.
    \end{lemma}

    \begin{remark}
        For comparison, if the agent is \emph{adaptively} calibrated with the same rate $r_\delta(\cdot)$, then \Cref{eq:blockwise-condition} can be satisfied by setting block length to be $r_\delta^{-1}(\gamma)$, making the total round complexity $T^{\mathrm{A}} = M\cdot r_\delta^{-1}(\gamma)$. In contrast, $\Alg^{\mathrm{NA}}$ has round complexity $T^{\mathrm{NA}} = \sum_{i=1}^M r_\delta^{-1}\!\left(\frac{\gamma}{2i}\right)$. Under the special case where $r_\delta(t)=O(t^{-\beta})$ for some $\beta>0$, we have $T^{\mathrm{A}} = O\!\left(M\gamma^{-1/\beta}\right)$, whereas $T^{\mathrm{NA}} = O\!\left(M^{1+1/\beta}\;\gamma^{-1/\beta}\right)=O\left(M^{1/\beta}\;T^{\mathrm{A}}\right)$.
    \end{remark}

    \proof{Proof of \Cref{lem:time-rescheduling}.}
The time-rescheduled algorithm $\Alg^{\mathrm{NA}}$ is defined by keeping the same block-boundary update rule as $\Alg$ (based on the most frequent agent action in each completed block), and modifying only the block lengths to
$L_i \triangleq r_\delta^{-1}\left(\frac{\gamma}{2i}\right)$.
Let $T_0\triangleq 0$, $T_i\triangleq \sum_{k=1}^i L_k$, and $I_i\triangleq \{T_{i-1}+1,\dots,T_i\}$.

We work on the event $\mathcal{E}$ (of probability at least $1-\delta$) on which the nonadaptive calibration guarantee holds simultaneously for all prefixes and all bins, i.e.,
$\calerr_j\left(\h_{1:t},\bp_{1:t}\right)\le r_\delta(t)$ for all $t\le T$ and all $j\in[\numA]$.
For any block $I_i=(T_{i-1},T_i]$ of length $|I_i|=L_i$, by triangle inequality and the nonadaptive calibration guarantee on event $\mathcal{E}$, we have
\begin{align*}
L_i\cdot \calerr_j\left(\h_{I_i},\bp_{I_i}\right)
={}& \left\|
\sum_{\tau=T_{i-1}+1}^{T_i} w_j(\bp_\tau)\left(\h_\tau-\bp_\tau\right)
\right\|_{\infty}\\
\le{}& \left\|
    \sum_{\tau=1}^{T_i} w_j(\bp_\tau)\left(\h_\tau-\bp_\tau\right)
    \right\|_{\infty}+\left\|
        \sum_{\tau=1}^{T_{i-1}} w_j(\bp_\tau)\left(\h_\tau-\bp_\tau\right)
        \right\|_{\infty}\\
={}& T_i\cdot \calerr_j\left(\h_{1:T_i},\bp_{1:T_i}\right)
+
T_{i-1}\cdot \calerr_j\left(\h_{1:T_{i-1}},\bp_{1:T_{i-1}}\right)\\
\le{}& T_i\cdot r_\delta(T_i)+ T_{i-1}\cdot r_\delta(T_{i-1}).
\end{align*}

Since $r_\delta$ is nonincreasing and $T_i\ge L_i$, we have
$r_\delta(T_i)\le r_\delta(L_i)\le \frac{\gamma}{2i}$ by the definition of $L_i$.
Moreover, the sequence $(L_i)_{i\ge 1}$ is nondecreasing, which implies $T_i=\sum_{k=1}^i L_k\le iL_i$. Combining the above inequalities gives
\[
T_i\cdot r_\delta(T_i)\le (iL_i)\cdot \frac{\gamma}{2i}=\frac{\gamma}{2}\cdot L_i.
\]
Similarly, we can show that $T_{i-1}\cdot r_\delta(T_{i-1})\le \frac{\gamma}{2}\cdot L_{i-1}\le \frac{\gamma}{2}\cdot L_i$.
As a result, we have
\[
    \calerr_j\left(\h_{I_i},\bp_{I_i}\right)\le\frac{T_i\cdot r_\delta(T_i)+ T_{i-1}\cdot r_\delta(T_{i-1})}{L_i}\le\gamma.
\]
Thus \eqref{eq:blockwise-condition} holds on every block under the new time schedule, and by the assumed correctness of $\Alg$ under \eqref{eq:blockwise-condition}, $\Alg^{\mathrm{NA}}$ succeeds on the event $\mathcal{E}$. Since $\Pr(\mathcal{E})\ge 1-\delta$, the success probability of $\Alg^{\mathrm{NA}}$ is at least $1-\delta$.
\Halmos
\endproof

Finally, we apply the reduction in \Cref{lem:time-rescheduling} to the principal's exploration algorithm (Algorithm~\ref{algo:principal-main}). Note that Algorithm~\ref{algo:principal-main} consists of three phases (Initialization, Optimization, and Robustness Check), each of which interacts with the agent through calls to the approximate membership oracle (Algorithm~\ref{algo:conservative-approx-MEM}), which admits a blockwise structure with each \emph{epoch} corresponding to a block in the sense of \Cref{lem:time-rescheduling}. Since the three phases use different error parameters when invoking the approximate membership oracle (corresponding to different values of the parameter $\gamma$ parameter in \Cref{eq:blockwise-condition}), we apply the time-rescheduling reduction to each of these phases separately. The resulting complexity of the time-rescheduled algorithm is characterized by the following corollary.
    
    \begin{corollary}[Time-rescheduled Algorithm~\ref{algo:principal-main} for non-adaptively calibrated agents]
        \label{cor:explore-nonadaptive}
        Suppose the conditions of \Cref{lemma:principal-main-complexity} hold, and the agent is \emph{nonadaptively} calibrated with rate $r_\delta(\cdot)=O(t^{-\beta})$ for some $\beta>0$.
        Then there exists a time-rescheduled version of Algorithm~\ref{algo:principal-main} that terminates within
        \begin{equation}
        \label{eq:principal-main-nonadaptive-rounds}
        O\!\left(
        \left(\frac{\sqrt{\numP}}{\vol(\eta/2)}\log^2 T\right)^{1+\frac{1}{\beta}}
        \cdot
        r_\delta^{-1}\!\left(
        \min\left\{\frac{\eta}{4\numA\numP},\frac{\epsrobust}{2\numA\sqrt{\numP}}\right\}
        \right)
        \right)
        \end{equation}
        rounds and, with probability at least $1-\delta-T^{-1}$, returns a pair $(\hat{\h},\hat{y})$ satisfying
        \textbf{(P1)} and \textbf{(P2)} from \Cref{lemma:principal-main-complexity}.
        \end{corollary}
        
 \begin{remark}[Comparison with no-regret agents]
    The time-rescheduling argument in \Cref{lem:time-rescheduling} relies on the fact that calibration error is defined as the norm of cumulative bias vector and therefore satisfies the triangle inequality. This allows guarantees on prefixes to be converted into guarantees on other sub-intervals. In contrast, external regret and swap regret are \emph{one-sided} quantities: they require the cumulative payoff not to exceed that of a comparator by too much, but place no restriction on how negative the difference can be. As a result, regret may be small on prefixes while being large on specific subintervals, which prevents us from obtaining blockwise guarantees through the same rescaling argument. Thus, the technique we used for calibration does not directly extend to regret-based notions.

    In addition, when combined with the exponential lower bound for learning against non-adaptively no-regret agents in \citet{brown2023learning}, our result highlights a separation between learning from a non-adaptively calibrated agent and non-adaptively no-regret agent. In the former setting,  polynomial-time learning algorithms are achievable via time rescheduling, whereas in the latter setting, there exists algorithm-independent exponential lower bounds for learning against non-adaptively no-external-regret agents.
\end{remark}

\subsection{Separation Between Adaptive and Non-Adaptive Agents}
\label{subsec:separation-adaptive-nonadaptive}
In this section, we provide a simple illustrative instance demonstrating that nonadaptive calibration may fail to provide timely information about best responses on intermediate subintervals. While not a formal separation result, this example highlights the challenges posed by nonadaptive calibration and illustrates why adaptivity leads to faster learning rates.
\begin{proposition}[Nonadaptive calibration can be uninformative on subintervals]
    \label{prop:example-separation}
    Fix $\beta\in(0,1]$ and a constant $C\ge \tfrac{1}{2}$.
    Let $r(t)\triangleq C\,t^{-\beta}$.
    There exist a repeated CSG with $m=2$ principal actions and $k=2$ agent actions, together with a sequence of principal strategies $\{\h_t\}_{t\ge 1}$ and agent forecasts $\{\bp_t\}_{t\ge 1}$, such that:
    
    \begin{enumerate}
        \item 
        For every $t\ge 1$, the nonadaptive calibration error for the agent's forecasts satisfies
        \[
            \max_{i\in\calA_A}\calerr_i(\h_{1:t},\bp_{1:t}) \;\le\; r(t).
        \]
        \item 
        For a subinterval of length $\Theta(s^{1-\beta})$, the adaptive calibration error satisfies
        \[
            \max_{i\in\calA_A}\calerr_i(\h_{s:t},\bp_{s:t}) \;\ge\; \Omega(1).
        \]
        \item The agent's action on the subinterval is different from the best response to the principal's strategy on the subinterval.
    \end{enumerate}
    \end{proposition}
    
    \proof{Proof of \Cref{prop:example-separation}.}
    We prove the proposition by constructing a simple $2\times 2$ game. 
    Let $A_P=\{a,b\}$ and $A_A=\{c,d\}$.
    Define the agent utility by
    \[
    U_A(a,c)=1,\qquad U_A(b,c)=-1,\qquad U_A(\cdot,d)=0,
    \]
    and assume deterministic tie-breaking in favor of action $d$.
    Then, for any prediction $\vp\in\calF_P$, it is not hard to see that
    $\BR(\vp) = \{d\}$ if $\vp_a > \tfrac12$, and $\BR(\vp) = \{c\}$ if $\vp_a \le \tfrac12$.
    
    Now we construct the principal strategy sequence and the agent forecasts.
    Fix any integer $s\ge 2$ and let $L \;\triangleq\; \lfloor s^{1-\beta}\rfloor$.
    Define the principal strategy sequence by
    \[
    \h_t \triangleq
    \begin{cases}
    (\tfrac12,\tfrac12) & \text{if } 1\le t \le s,\\
    (1,0) & \text{if } s< t \le s+L,\\
    (\tfrac12,\tfrac12) & \text{if } t > s+L.
    \end{cases}
    \]
    On the other hand, the agent's forecasts are constant on the entire horizon, i.e., 
    \[
    \bp_t \equiv \left(\tfrac12,\tfrac12\right)\qquad\forall t\ge 1.
    \]
    Since $\bp_t$ never changes, the agent always plays action $d$. 
    This shows that the agent's action during $(s,s+L]$ is different from the best response to the principal's strategy, which is action $c$. This proves item~(3). Item~(2) also follows immediately as the adaptive calibration error on the subinterval $(s,s+L]$ is at least $\frac{1}{2}$.
    
    We now bound the non-adaptive calibration error for the agent's forecasts.
    Since action $c$ is never played, it suffices to bound the calibration error for action $d$. For $t\le s$, $\h_\tau=\bp_\tau$ for all $\tau\le t$, so the error is $0$.
    For $t>s$, we have
    \begin{align*}
    \calerr_d(\h_{1:t},\bp_{1:t})
    =\frac{\min(t-s,L)}{t}\cdot\left\|(1,0)-(\tfrac12,\tfrac12)\right\|_\infty\le\frac{L}{2s}=\frac{s^{-\beta}}{2}\le \frac{t^{-\beta}}{2}.
    \end{align*}
    This proves item~(1) with $r(t)=C\cdot t^{-\beta}$ for constant $C\ge \tfrac12$. The proof is complete.
    \Halmos
    \endproof
    
    In this example, if the agent were \emph{adaptively} calibrated with the same rate $r(t)=C\cdot t^{-\beta}$, then for sufficiently large $s$, the adaptive calibration guarantee would force the agent to play action $c$ on at least a majority of time steps in the subinterval $(s,s+L]$. This would allow the principal to gain information about the agent's utility function. In contrast, nonadaptive calibration does not exclude the trivial agent strategy of always forecasting $\bp_t\equiv (\tfrac12,\tfrac12)$ and responding with action $d$ on the entire horizon, effectively providing no useful information about the agent's utility function.

%% file: conclusion.tex
\section{Discussion and Future Directions}\label{sec:conclusion}

In this paper, we introduced and studied Calibrated Stackelberg Games (CSGs), a generalization of the classic Stackelberg framework in which agents no longer have perfect knowledge of the principal's actions, but instead best respond to calibrated forecasts. Our results demonstrate that despite this informational asymmetry, the principal's optimal utility converges to the classical Stackelberg value $V^{\star}$, both in discrete and continuous settings. This is a somewhat surprising finding: even when both players (i.e., the principal and the agent) have significantly less information than in standard models, their repeated play converges to the outcome of a Stackelberg equilibrium. On the agent side, we introduced the notion of adaptive calibration in games, showing that it arises naturally from online learning algorithms and providing a constructive approach for efficiently generating adaptively calibrated forecasts. Taken together, these contributions advance both the theory of Stackelberg interactions and the broader understanding of calibration as a behavioral foundation in strategic learning settings.

There are several exciting avenues for future research, some of which we highlight below.

\paragraph{Knowledge of calibration error $r_\delta(\cdot)$.} In our work, the principal's learning algorithm (Algorithm~\ref{algo:principal-main}) requires access to the agent's calibration rate $r_\delta(\cdot)$ or an upper bound of the rate. This is because ---although the actual forecasting algorithm that the agent uses is unknown to the principal--- \emph{some} information regarding how $\vp_t$'s relate to $\h_t$'s is necessary to leverage the fact that agents are calibrated. But we think that in some specific settings (e.g., strategic classification) there may actually exist \emph{extra} information regarding the forecasts (compared to just knowing $r_\delta(\cdot)$) that can be leveraged to design learning algorithms for the principal with faster convergence rates. For example, in strategic classification, there may exist correlations between agent features that could be leveraged to reveal (an upper bound of) $r_\delta(\cdot)$. 

\paragraph{$\ell_\infty$ calibration.} Although the results in this paper are all stated in terms of the $\ell_\infty$-calibration error (i.e., the \emph{maximum} error over binning functions), a lot of the existing calibration literature focuses on $\ell_1$-calibration error~\citep{foster1997calibrated,foster1998asymptotic} (i.e., the \emph{sum} of errors over binning functions). It is an interesting problem whether we can get $\ell_1$-adaptive calibration error bounds without a polynomial dependency in the number of binning functions, where obtaining such bounds lead to polynomial improvements on the dependency of $\numP$ (the number of agent's actions). In the case of continuous calibration, it is an open problem to obtain uniform (adaptive) calibration error bounds for parametric or nonparametric continuous binning function classes. Resolving this open problem could lead to a better rate for the learning direction of \Cref{thm:main-continuous}, as the current result uses naive $\ell_\infty$-to-$\ell_1$ conversion of calibration error that leads to linear dependency on the number of binning functions, which turns out to be exponential in the dimension of the principal's action space.
See Remark~\ref{remark:l1-calibration} for more details.

%% file: app-calibration-std-def.tex
{\bf \Huge Appendix}

\section{Calibrated Forecasts Standard Definition}\label{app:calibration-std}

We give below the standard definition for asymptotic calibration of \citet{foster1998asymptotic} for a sequence of binary outcomes, i.e., $\h_t \in A=\{0,1\}, \forall t \in [T]$. The forecasts $\vp_t$ take values in $C=[0,1]$. Let $X$ denote the adaptive adversary generating the events' sequence (which is of infinite size), where the $T$ {first} events are $\h_1, \dots, \h_T$.
\begin{definition}
    A forecasting procedure $\sigma$ is \emph{asymptotically calibrated} if and only if for any adaptive adversary $X$ that generates the sequence $\h_1,\cdots,\h_T\in A$ and the forecasting algorithm $\sigma$ that generates (possibly random) forecasts $\vp_1,\cdots,\vp_T\in C$ on the same sequence, we have that the calibration score $C_T(X,\sigma)$ goes to $0$ as $T \to \infty$:
    \begin{align}
        C_T(X,\sigma)=\sum_{\vp\in\calF_P}
        \frac{n_T(\vp;\h,\sigma)}{T}
        \left|\rho_T(\vp;\h,\sigma)-\vp\right|
        \label{eq:def-cal-standard}
    \end{align}
    where $n_T(\vp;\h,\sigma)\triangleq\sum_{t \in [T]} \indicator{\vp_t=\vp}$ is the number of times that $\sigma$ predicts $\vp$ and $\rho_T(\vp;\h,\sigma)\triangleq\frac{\sum_{ t \in [T]} \h_t\indicator{\vp_t=\vp}}{n_T(\vp;\h,\sigma)}$ be the fraction (empirical probability) of these times that the actual event was $1$. 
\end{definition}

Note that in Eq.~\eqref{eq:def-cal-standard}, while $\calF_P$ contains an infinite number of distinct $\bp$'s (hence an infinite number of summands), for every finite $T$, there is only a finite number of $\bp$ where $n_T(\vp;\h,\sigma)$ is nonzero. Therefore, $C_T$ is well-defined and finite.

Equivalently, the above definition states that for the infinite binning~\cite{foster2021forecast} $\Pi=\big\{w_{x}(\bp):x\in C\big\}$ where $w_x(\bp)=\indicator{\bp=x}$, the calibration score can be equivalently expressed as
    \begin{align*}
        C_T(X,\sigma)\triangleq \sum_{w_{x}\in \Pi}\frac{n_T(x)}{T}\left\|\overline{\h}_{T}(x)-\overline{\bp}_{T}(x)\right\|,
    \end{align*}
    where  $n_{T}(x)\triangleq\sum_{t=1}^T w_x(\bp_t)$ is the number of times that forecast $\bp_t$ falls into bin $x$, $\bar{\bp}_{T}(x) \triangleq\sum_{t=1}^T \frac{w_x(\bp_t)}{n_{T}(x)}\cdot \bp_t$ is the average forecast that activates bin $x$, which is equal to $\sum_{t=1}^T \frac{w_x(\bp_t)}{n_{T}(x)}\cdot x=x$ because $w_x(\bp_t)$ is nonzero if and only if $\bp_t=x$, and $\bar{\h}_{T}(x) \triangleq\sum_{t=1}^T \frac{w_x(\bp_t)}{n_{T}(x)}\cdot {\h_{t}}$ is the average outcome corresponding to bin $x$. It follows that the score $C_T$ is a sum of the calibration errors during interval $[1:T]$ for all bins (with $\calerr$ defined in \Cref{def:calibration-adaptive}).
    \begin{align*}
        C_T(X,\sigma)=\sum_{w_x\in\Pi}\calerr_x(\h_{1:T},\bp_{1:T}).
    \end{align*}

%% file: app-swap-regr.tex
\section{Calibrated Forecasts Lead to No Swap Regret}\label{app:swap-regret}

In this section, we show the connection between no-swap-regret agents and adaptively calibrated ones. As a reminder, no-swap-regret agents (translated to our setting and notation for the ease of exposition) are defined as follows.

\begin{definition}[Agent's swap regret~\citep{blum2007external}]
    For a sequence of principal's strategies $\h_1,\cdots,\h_T\in\cH_P$ and agent's actions $y_1,\cdots,y_T\in\actionA$, the swap regret is defined as
    \begin{align*}
        \swapreg(\h_{1:T},y_{1:T})=\max_{\pi:\actionA\to\actionA}\sum_{t \in [T]} U_A(\h_t, \pi(y_t))-\sum_{t \in [T]} U_A(\h_t, y_t).
    \end{align*}
We say that an agent is a \emph{no-swap-regret} agent, if for the sequence of actions $\{y_t\}_{t \in [T]}$ that they are playing it holds that $\swapreg(\h_{1:T},y_{1:T}) = o(T)$.
\end{definition}

We next show that calibrated forecasts lead to no swap regret actions for the agent. 

\begin{lemma}[Calibrated forecasts lead to no swap regret]
\label{lemma:no-swap-regret}
    If the agent is $(\eps, \Pi)$-adaptively calibrated, then the agent's swap regret on the sequence $\h_{1:T}$ is bounded by the calibration error as follows:
    \begin{itemize}
        \item If the agent breaks ties deterministically, then with probability $\ge1-\delta,$
        \begin{align*}
            \swapreg(\h_{1:T},y_{1:T})\le 2U_{\max}\numP \numA T\left(r_\delta(T)+\eps\right)\in o(T).
        \end{align*}
        \item If the agent breaks ties randomly, then with probability $\ge1-2\delta$,
        \begin{align*}
            \swapreg(\h_{1:T},y_{1:T})\le U_{\max}\left(O\left(\sqrt{T\numA \log\left(\frac{\numA}{\delta}\right)}\right)+2\numP\numA T\left(r_\delta(T)+\eps\right)\right)\in o(T).
        \end{align*}
    \end{itemize}
    where $\Umax = \max_{\h \in \calA_P} \max_{y \in \calA_A} U_A(\h, y)$ is the maximum utility the agent can obtain (without constraining the agent to play best responses).
\end{lemma}
\proof{Proof of Lemma~\ref{lemma:no-swap-regret}.} We first present the proof for the case that the agents break ties deterministically. To simplify notation, we use $n_T(i) := n_{[0:T]}(i)$, $\bar{\vp}_T := \bar{\vp}_{[0:T]}(i)$, and $\bar{\h}_T(i) := \bar{\h}_{[0:T]}(i)$.

Fix a $\pi:\actionA\to\actionA$. Then, with probability at least $1 - \delta$, we have that:
\begin{align*}
        &\sum_{t=1}^T U_A(\h_t, \pi(y_t))-\sum_{t=1}^T U_A(\h_t, y_t)\\
       &=\sum_{i\in\actionA}\sum_{t=1}^T \indicator{y_t=i} \left(U_A(\h_t, \pi(i))- U_A(\h_t, i)\right) &\tag{rewriting $y_t$ as the exact action}\\
        &\stepa{=}\sum_{i\in\actionA}\sum_{t=1}^T w_i(\bp_t) \left(
        \left\langle \h_t, U_A(\cdot,\pi(i))\right\rangle 
        -\left\langle\h_t, U_A(\cdot,i)\right\rangle\right)&\numberthis{\label{eq:bef-main-proof}}\\
        &=\sum_{i\in\actionA} n_T(i) \left(
        \left\langle \bar{\h}_T(i), U_A(\cdot,\pi(i))\right\rangle 
        -\left\langle\bar{\h}_T(i), U_A(\cdot,i)\right\rangle\right)\\
        &=\sum_{i\in\actionA}n_T(i) \Bigg({\left\langle 
        \bar{\h}_T(i)-\bar{\bp}_T(i), U_A(\cdot,\pi(i))\right\rangle}
        +\left\langle \bar{\bp}_T(i), U_A(\cdot,\pi(i))-U_A(\cdot,i)
        \right\rangle\\
        &\qquad\qquad\qquad\qquad\qquad\qquad
        +{\left\langle \bar{\bp}_T(i)-\bar{\h}_T(i), U_A(\cdot,i)\right\rangle}\Bigg)\\
        &\stepb{\le} \sum_{i\in\actionA}n_T(i) \left\|\bar{\bp}_T(i)-\bar{\h}_T(i)\right\|_{\infty}\left(\|U_A(\cdot,\pi(i))\|_1+\|U_A(\cdot,i)\|_1\right)\\
        &=2U_{\max}\numP \cdot \sum_{i\in\actionA} T\cdot\calerr_i(\h_{1:T},\bp_{1:T}) &\tag{Def.~\ref{def:calibration-adaptive}}\\
        & \le 2U_{\max}\numP \numA T\left(r_\delta(T)+\eps\right).
    \end{align*}
    In the above equations, step (a) is due to the fact that agents best respond with a deterministic tie-breaking rule: $y_t=i$ if and only if $i\in\BR(\bp_t)$ and $i\succ j,\forall j\neq i$, which is equivalent to $w_i(\bp_t)=1$. We have also used $U_A(\cdot,i)$ to denote the $\numP$-dimensional vector where the $j$th entry is the utility $U_A(j,i)$.
    Step (b) is because the second term 
    \[\left\langle \bar{\bp}_T(i), U_A(\cdot,\pi(i))-U_A(\cdot,i) \right\rangle=U_A(\bar{\bp}_T(i),\pi(i))-U_A(\bar{\bp}_T(i),i)\] is non-positive since each $\bp_t$ with $w_i(\bp_t)=1$ belongs to the best response polytope $\polytope_i$, so does their average: $\bar{\bp}_t(i)\in \polytope_i\iff i\in \BR(\bar{\bp}_t(i))$.

        Since the above inequality holds for any $\pi$, it also holds after taking the maximum over all $\pi:\actionA\to\actionA$. Therefore, we have the same bound for the agent's swap regret.

    Next, we move to the case when the agent breaks ties randomly. For a fixed $\pi$, we have that at every time step $t$, 
    \begin{align*}
        \mathbb{E}_{t-1}\left[
            U_A(\h_t,\pi(y_t))
        \right]
        =\frac{\sum_{i\in\BR(\bp_t)}U_A(\h_t,\pi(i))}{|\BR(\bp_t)|}
        =\sum_{i\in\actionA} w_i(\bp_t) U_A(\h_t,\pi(i)).
    \end{align*}
    Therefore, by Azuma-Hoeffding's inequality, w.p. $\ge 1-\delta'$, we have
    \begin{align}\label{eq:azuma-pi}
        \sum_{t=1}^T U_A(\h_t,\pi(y_t)) \le \sum_{t=1}^T w_i(\bp_t)\sum_{i\in\actionA} U_A(\h_t,\pi(i)) + O\left(\sqrt{ T\log\left(\frac{1}{\delta'}\right) }\right).
    \end{align}
    Since all actions in $\BR(\bp_t)$ have the same utility for the agents, we also have
    \begin{align*}
        U_A(\h_t,y_t)=\frac{\sum_{i\in\BR(\bp_t)} U_A(\h_t,i)}{|\BR(\bp_t)|}
        =\sum_{i\in\actionA} w_i(\bp_t) U_A(\h_t,i).
    \end{align*}
    Therefore, using Equations~\eqref{eq:azuma-pi} and~\eqref{eq:azuma-br}, we have that with probability at least $1-\delta'$,
    \begin{align*}
        \sum_{t=1}^T U_A(\h_t, \pi(y_t))-\sum_{t=1}^T U_A(\h_t, y_t)
        &\le \sum_{i\in\actionA}\sum_{t=1}^T w_i(\bp_t) \left(U_A(\h_t,\pi(i))-U_A(\h_t,i)\right) \\
        &\quad \quad \quad +O\left(\sqrt{ T\log\left(\frac{1}{\delta'}\right) }\right) \numberthis{\label{eq:azuma-br}}.
    \end{align*}
    We can use the same arguments as above (from Equation~\eqref{eq:bef-main-proof} onwards) to bound the first term on the right hand side by $2U_{\max}\numP \numA T\left(r_\delta(T)+\eps\right)$ with probability $1-\delta$. Finally, setting $\delta'=\delta/M$ where $M=\numA^{\numA}$ is the number of possible swap functions, and applying the union bound, we conclude that with probability $\ge1-2\delta$, the swap regret is bounded by
    \begin{align*}
        \swapreg(\h_{1:T},y_{1:T})\le U_{\max}\left({O} \left(\sqrt{T \left(\numA \log\numA+\log\left(\frac{1}{\delta}\right)\right)}\right)+2\numP\numA T\left(r_\delta(T)+\eps\right)\right).    \end{align*}
        The proof is complete.
\Halmos
\endproof

%% file: app-learning-algos.tex
\section{Supplementary Material for Section~\ref{sec:learning-algos}}\label{app:learning-algos}

\subsection{Approximate membership oracle to conservative polytopes}
\label{sec:approx-mem}

In this section, we formally present the algorithm (Algorithm~\ref{algo:conservative-approx-MEM}) for constructing an approximate membership oracle to the conservative best response polytope for each of the agent's action. The sample complexity of the oracle will be presented in \Cref{lem:apx-memb-oracle}.

\begin{algorithm}[htbp]
    \caption{$\approxmem$: Approximate membership oracle for conservative polytopes}
    \label{algo:conservative-approx-MEM}
    \DontPrintSemicolon
    \LinesNumbered
    \SetAlgoNoLine
    \SetAlgoNoEnd
    \KwIn{queried strategy $\h$, approximation factor $\eps_1$, conservatism factor $\eps_2$, failure probability $\eps_3$;
    If True/False mode, queried action $y\in\actionA$.}
    \KwOut{With probability $1-\eps_3$, return \textsc{True} if $\h\in \ball_2(\polytope_y^{-\eps_2},-\eps_1)$ and \textsc{False} if $\h\notin \ball_2(\polytope_y^{-\eps_2},{-\frac{\eps_1}{2\sqrt{\numP}}})$.}
    \textbf{Parameters: }Number of epochs $\ntrials=10\sqrt{\numP}\log(\frac{1}{\eps_3})$, radius $R=\eps_1$, 
    number of samples per epoch $\nrepeat=r_\delta^{-1}(\frac{\eps_2}{\numA\sqrt{\numP}})$.\;

    \For{epoch $\phi \in [\Phi]$}{ 
        Sample a point $\h_\phi$ such that $\|\h_\phi-\h\|_2=R$.
        \tcc*{Sample point $\h_{\phi}$ close to $\h$}
        \textbf{If} {$\h_\phi\not\in\cH_P$} \textbf{then} {\textsc{Return} \textsc{False} \tcc*{If $\h_\phi$ is not a feasible principal strategy, return False.}}
        \Else{Play strategy $\h_\phi$ for $\nrepeat$ rounds.\;
        $y_\phi$ $\gets$ most frequent best-response action from agent during the $\nrepeat$ rounds.\;
        {\bf if} $y_\phi\neq y$ {\bf then} \textsc{Return False} \tcc*{
        $\h_\phi$ is too close to $\polytope_{y_\phi}$}
        }}
        \textsc{Return} \textsc{True}\tcp{For membership, output $y_\Phi$ if all $\{y_\phi\}_{\phi\in[\ntrials]}$ agree, and $\perp$ otherwise}
        
    \end{algorithm}

In \Cref{lem:apx-memb-oracle}, we show that the parameters $\ntrials,\epscal,R$ can be tuned to achieve a wide range of parameters $(\eps_1,\eps_2,\eps_3) $.

\begin{lemma}[Approximate membership oracle]
    \label[lemma]{lem:apx-memb-oracle}
    Suppose $\eps_1,\eps_2$ are given parameters 
    {that satisfy $\eps_2\le \frac{\eps_1}{\sqrt{\numP}}$.
    Let the agent be $(\eps,\Pi)$-adaptively calibrated with rate $r_\delta(\cdot)$ and infinitesimal $\eps$.}
    
    Then with parameters $\ntrials,R,\nrepeat$ that satisfy
    {$R=\eps_1$, 
    
    $\ntrials=100\sqrt{\numP}
    \cdot\log(\frac{1}{\eps_3})$,
    and $\nrepeat=r_\delta^{-1}(\frac{\eps_2}{\numA\sqrt{\numP}})$,} then conditioned on the success event of agent's calibration algorithm,
    Algorithm~\ref{algo:conservative-approx-MEM} (\approxmem) returns an $\eps_1$-approximate membership oracle to $\polytope_y^{-\eps_2}=B_2(\polytope_y,-\eps_2)$ with probability $1-\eps_3$, using no more than $N_{\eps_1,\eps_2,\eps_3}=O\left(\ntrials \cdot \nrepeat\right)={O}\left(\sqrt{\numP}\cdot r_\delta^{-1}(\frac{\eps_2}{\numA\sqrt{\numP}})\log(\frac{1}{\eps_3})\right)$ rounds of interactions with the agent. 
    
    Specifically, with probability $1-\eps_3$, $\approxmem$ either returns $\textsc{True}$ which asserts that $\h\in B_2(\polytope_y^{-\eps_2},{-\frac{\eps_1}{\sqrt{\numP}}})$, or returns $\textsc{False}$ which asserts that $\h\notin B_2(\polytope_y^{-\eps_2},-\eps_1)$.
\end{lemma}

\proof{Proof of \Cref{lem:apx-memb-oracle}.}
   Before we delve into the specifics of the proof, we introduce some notation. In order to make sure that $\h_\phi$ is such that $\|\h - \h_\phi\| = R$, we do the following: $\h_\phi \gets \h + R \vcs_\phi$, where $\vcs_\phi$ is sampled uniformly at random from the equator $\mathbb{S}\cap\mathbb{H}$, where $\bs=\{\vs\in\br^{\numP}:\|\vs\|_2^2=1\}$ is the unit sphere and $\mathbb{H}=\{\vs\in\br^{\numP}:\langle \vs,\mathbf{1}\rangle =0\}$ is an equatorial hyperplane ($\mathbf{1}\triangleq(1,\cdots,1)\in\br^{m}$). Note that this is because we want $\h_\phi$ to remain a valid probability distribution, i.e., that $\langle\h_\phi,\mathbf{1}\rangle = 1$ and $\h_\phi\ge0$ coordinate-wise; indeed, since we already have $\langle \h,\mathbf{1}\rangle = 1$, we need to make sure that (1) $\langle \vcs_\phi,\mathbf{1}\rangle = 0$, which is guaranteed by $\vcs_\phi\in\mathbb{H}$; (2) $\h_\phi\ge0$, which is guaranteed by returning $\textsc{False}$ whenever $\h_\phi\notin\calH_P$.
    
   For the rest of the proof, we condition on the following success event $\cE$, {which is the event that the agent has bounded calibration error rate $r_\delta(t-s)$ in every window of time $[s,t]$}:
   \[\mathcal{E}\triangleq
   \Big\{
    \forall [s,t]\subseteq [1,T],\ \forall i\in\actionA,\;
    \calerr_i\left(\h_{s:t},\bp_{s:t}\right)\le r_\delta(t-s)
   \Big\}.
   \]

    We first show that conditioned on $\cE$, we have $\h_\phi\in B_2(\polytope_{y_\phi},\eps_2)$ for all $\phi\in[\ntrials]$. 
    Let $\nrepeat_{y_\phi}$ be the number of times that agent plays $y_\phi$ during the $\nrepeat$ repeats, then we have $\nrepeat_{y_\phi}\ge \nrepeat/\numA$ because $y_\phi$ is the most frequently played action. 
    
    Recall that in Algorithm~\ref{algo:conservative-approx-MEM}, {we have chosen the number of samples in each epoch $\nrepeat$ to be such that $r_\delta(\nrepeat)+\eps=\frac{\eps_2}{\numA\sqrt{\numP}}$.}
    Then, the calibration error bound in \Cref{def:calibration-adaptive} guarantees that
    \begin{align}
        &\frac{\nrepeat_{y_{\phi}}}{\nrepeat}\left\|\Bar{\bp}(y_{\phi})-\h_{\phi}\right\|_\infty
        =\calerr_{y_{\phi}}(\h_{\phi,1:\nrepeat},\vp_{\phi,1:\nrepeat})
        \le r_\delta(\nrepeat)+\eps
        =\frac{\eps_2}{\numA\sqrt{\numP}},
        \nonumber\\
        \Rightarrow\; &\left\|\Bar{\bp}(y_{\phi})-\h_{\phi}\right\|_2\le\sqrt{\numP}\left\|\Bar{\bp}(y_{\phi})-\h_{\phi}\right\|_{\infty}
        \le \sqrt{\numP}\numA \cdot \frac{\eps_2}{\numA\sqrt{\numP}}=\eps_2.
        \label{eq:hphi-ball}
    \end{align}
where the first inequality in equation~\eqref{eq:hphi-ball} is because of the norm property $\|x\|_2 \leq \sqrt{d} \|x \|_\infty$ for a vector $x \in \bbR^d$. Since $\Bar{\bp}(y_{\phi})\in\polytope_{y_\phi}$ because the agent always best responds to forecasts, we obtain $\h_\phi\in B_2(\polytope_{y_\phi},\eps_2).$
    
    We then prove the following two claims:
    \squishlist
    \item[\Cone] If $\h\in B_2(\polytope_y^{-\eps_2},-\eps_1)$, then $\approxmem$ returns \textsc{True} with probability $1$.
    \item[\Ctwo] If $\h\notin B_2(\polytope_y^{-\eps_2},-\frac{\eps_1}{\sqrt{\numP}})$, then $\approxmem$ returns \textsc{False} with probability $\ge1-\eps_3$.
    \squishend
    Indeed, if the following two claims hold, then we have established that $\approxmem$ asserts one of two cases correctly with probability $\ge1-\eps_3$ conditioned on $\cE$.

    \textbf{Proof of \Cone.} Suppose $\h\in B_2(\polytope_y^{-\eps_2},-\eps_1)$, which equivalently implies that $\h\in B_2(\polytope_y,-(\eps_1+\eps_2))$. Therefore, the distance between $\h$ and any other polytope $\polytope_{y'}\ (y'\neq y)$ must be larger than $\eps_1+\eps_2$. By triangle inequality, for any other strategy $\h'$ in a different polytope $\polytope_{y'}$ (i.e., $\h'\in\polytope_{y'}$), we have
    \[
        \forall\phi\in[\Phi],\quad\|\h_\phi-\h'\|_{2}\ge\|\h-\h'\|_2-\|\h-\h_{\phi}\|_2
        >\eps_1+\eps_2-R=\eps_2.
    \]

    Since this holds for all $\h'\in\polytope_{y'}$, it implies $\h_{\phi}\not\in B_2(P_{y'},\eps_2)$ whenever $y'\neq y$.
    Together with the fact that $\h_\phi\in B_2(\polytope_{y_\phi},\eps_2),$ we must have $y_\phi=y$ for all epochs $\phi\in[\Phi]$. Therefore, $\approxmem$ always returns $\textsc{True}$.

    \textbf{Proof of \Ctwo.} Suppose $\h\notin B_2(\polytope_y^{-\eps_2},{-\frac{\eps_1}{\sqrt{\numP}}})$.
    We first analyze the probability of returning \textsc{False} for a fixed epoch $\phi\in[\Phi]$ by showing that $\h_\phi\not\in B_2(\polytope_y,\eps_2)$ with high probability.

    If $\h\in B_2(\polytope_y,\eps_2)$, then by triangle inequality, the distance between $\h$ and the boundary of $B_2(\polytope_y,\eps_2)$ must be no more than $2\eps_2+\frac{\eps_1}{\sqrt{\numP}}\le\frac{3\eps_1}{\sqrt{\numP}}$.
    Since $\h_\phi$ is uniformly sampled from the sphere of radius $R=\eps_1$ around $\h$,
    by convexity of $B_2(\polytope_y,\eps_2)$ and the rotation invariance property of a unit sphere,
    we have
    \[
        \Pr[\h_\phi\notin B_2(\polytope_y,\eps_2)]
        \ge \Pr\left[\langle R\vcs_\phi, \mathbf{v}\rangle\ge \frac{3\eps_1}{\sqrt{\numP}}\right]
        =\Pr\left[\langle \vcs_\phi, \mathbf{e}_1\rangle\ge\frac{3}{\sqrt{\numP}}\right],
    \]
    where $\mathbf{v}$ is a unit vector pointing in the direction of (one of) the projection from $\h$ to the boundary of $B_2(\polytope_y,\eps_2)$, and $\mathbf{e}_1=(1,0,\cdots,0)\in\br^{\numP}$.
    According to \citep[Lemma 9]{feige2002optimality}, we can further lower bound the probability by
    \[
    \Pr[\h_\phi\notin B_2(\polytope_y,\eps_2)]\ge 
    \frac{1}{2\sqrt{\numP}}\left(1-\left(\frac{3}{\sqrt{\numP}}\right)^2\right)^\frac{\numP-1}{2}\ge\Omega\left(\frac{1}{\sqrt{\numP}}\right).
    \]
    Finally, to make the probability that no epoch returns \textsc{False} (failure of $\approxmem$) at most $\eps_3$, we need
    \[
    (1-\Pr[\h_\phi\notin B_2(\polytope_y,\eps_2)])^{\Phi}\le\eps_3
    \quad\iff\quad
    \Phi\ge\frac{\log(\frac{1}{\eps_3})}{\Pr[\h_\phi\notin B_2(\polytope_y,\eps_2)]},
    \]
    which is satisfied by our choice of parameter $\Phi=\Theta\left(\sqrt{\numP}\log(\frac{1}{\eps_3})\right)$.
    Therefore, whenever $\h\notin B_2(\polytope_y^{-\eps_2},-\frac{\eps_1}{\sqrt{\numP}})$, we have proved that $\approxmem$ returns \textsc{False} with probability $\ge1-\eps_3$.
\Halmos
\endproof

\subsection{Proof of the Guarantee of Algorithm~\ref{algo:principal-main}}
\label{app:proof-principal-main-complexity}

\lemmaPrincipalMainComplexity*

This section is devoted to the proof of \Cref{lemma:principal-main-complexity}.
We prove the three main components of Algorithm~\ref{algo:principal-main} separately: 
\begin{itemize}
    \item In \Cref{lemma:initialization-set}, we analyze the \textbf{Initialization Phase} of Algorithm~\ref{algo:principal-main}, showing that the constructed initialization set $\calI$ contains a pair $(\h_0,y_0)$ such that $y_0=\yst$  is the agent's strategy in the Stackelberg equilibrium, and $\h_0$ is well-centered in the polytope $\polytope_{\yst}$. This strategy serves as a good initial point for the optimization phase and will subsequently guarantee the existence of a near-optimal solution in the optimization phase.
    \item In \Cref{lemma:optimization-phase}, we analyze the \textbf{Optimization Phase} of Algorithm~\ref{algo:principal-main} by invoking a lemma from \citet{haghtalab2022learning} for optimizing affine functions using membership oracles, showing that when the initialization point is well-centered in the polytope for $\yst$, the optimization phase will find a strategy that is both near-optimal and robustly within the target polytope.
    \item In \Cref{lemma:robustness-check}, we analyze the \textbf{Robustness Check Step} of Algorithm~\ref{algo:principal-main}. The goal of this step is to verify that the strategy found by the optimization phase is robustly within the corresponding best-response polytope. This step is necessary since the robustness guarantee of the optimization phase is conditioned on the initialization point being $\frac{\eta}{2}$-robust in target polytope, however, since the membership oracle $\approxmem$ is only approximate, the initialization set may also contain points that are $\frac{\eta}{2\sqrt{\numP}}$-robust but not $\frac{\eta}{2}$-robust. This discrepancy can cause the optimization phase to produce a strategy that lacks sufficient robustness. Therefore, the robustness check step is designed to filter out these insufficiently robust points.
\end{itemize}

\begin{lemma}[Initialization Phase of Algorithm~\ref{algo:principal-main}]
    \label[lemma]{lemma:initialization-set}
    Under the regularity assumption in Assumption~\ref{assumption:regularity},
    with probability at least $1-\frac{1}{5T}$, the initialization set $\calI$ contains $(\h_0,y_0)$ where $y_0=\yst$ is the optimal target, and $\h_0$ is $\frac{\eta}{2}$-centered in $\polytope_{\yst}$, i.e., $\h_0\in B_2(\polytope_{\yst},-\frac{\eta}{2})$.

    The total number of samples required for the initialization phase is $O\left(
        \frac{\sqrt{\numP}}{\vol(\eta/2)}
        \cdot r_\delta^{-1}(\frac{\eta}{4\numA{\numP}})\log^2 T
    \right)$.

\end{lemma}
\proof{Proof of \Cref{lemma:initialization-set}}.
    By regularity assumption (Assumption~\ref{assumption:regularity}), there exists a well-centered strategy $\dot{\h}\in\polytope_{\yst}$, s.t. $B_2(\dot{\h},\eta)\in\polytope_{\yst}$. Therefore, for all strategies $\h'\in B_2(\dot{\h},\frac{\eta}{2})$, we have $\h'\in B_2(\polytope_{\yst},-\frac{\eta}{2})$, i.e., $\h'$ lies robustly inside $\polytope_{\yst}$. Since the set of all such $\h'$ takes up nontrivial volume $\vol(\eta/2)$ in $\calH_P$, we know that $O(V^{-1}\log T)$ uniform samples are guaranteed to hit one with probability $1-\frac{1}{2T}$. In the rest of the proof, we will show that with probability $1-\frac{1}{10T}$, once such an $\h'\in B_2(\polytope_{\yst},-\frac{\eta}{2})$ is found, $\approxmem$ will return the correct membership $\yst$ and add $(\h',\yst)$ to the initialization set $\calI$. Together, they imply that $\calI$ must contain one such pair with probability at least $1-\frac{1}{5T}$.
    
    On the success event of $\approxmem(\eps_1=\frac{\eta}{4},\eps_2=\frac{\eta}{4\sqrt{\numP}},\eps_3=\frac{1}{10T^2})$ --- which holds with probability $1-\frac{1}{10T}$ after a union bound over no more than $T$ queries --- \Cref{lem:apx-memb-oracle} guarantees that $\approxmem(\h_i)$ must fall into one of the two cases for each queried strategy $\h_i$:
    \begin{enumerate}
        \item It returns $\perp$, which asserts that $\h_i\not\in B_2(\polytope_y, -\frac{\eta}{4}(1+\frac{1}{\sqrt{\numP}}))$ for any $y\in\BR(\h_i)$;
        \item It returns $y_i\neq\perp$, which asserts that $\h_i\in B_2(\polytope_{y_i},-\frac{\eta}{2\sqrt{\numP}}))$.
    \end{enumerate}

    We will now show that if a queried strategy $\h_i$ is in $B_2(\polytope_{\yst},-\frac{\eta}{2})$, then $\approxmem$ must falls into the second case, and successfully return membership $\yst$. Indeed, for the sake of contradiction, if $\approxmem$ falls into the first case, it must be that $\h_i\not\in B_2(\polytope_{\yst},-\frac{\eta}{4}(1+\frac{1}{\sqrt{\numP}}))$, which implies that $\h_i\not\in B_2(\polytope_{\yst},-\frac{\eta}{2})$, a contradiction. Therefore, $\approxmem$ must return membership $\yst$, because $\yst$ is the only action that satisfies $\h_i\in B_2(\polytope_{\yst},-\frac{\eta}{2\sqrt{\numP}})$.

    We have thus proved that with probability $1-\frac{1}{5T}$, the initialization set $\calI$ contains $(\h_0,y_0)$ where $y_0=\yst$ is the optimal target, and $\h_0\in B_2(\polytope_{\yst},-\frac{\eta}{2})$. It remains to bound the sample complexity of the initialization phase.
According to \Cref{lem:apx-memb-oracle}, each call to $\approxmem$ takes $O\left(\sqrt{\numP} r_\delta^{-1}(\frac{\eta}{4\numA{\numP}})\log T\right)$ samples, and the initialization phase calls $\approxmem$ for $\Ninit=O(\frac{\log T}{\vol(\eta/2)})$ times, the total sample complexity is $O\left(\frac{\sqrt{\numP}}{\vol(\eta/2)} r_\delta^{-1}(\frac{\eta}{4\numA{\numP}})\log^2 T\right)$, as desired.
\Halmos
\endproof

\begin{lemma}[Lemma F.4 in \citep{haghtalab2022learning} Repurposed for our setting]
    \label{lemma:optimization-phase}
    Fix an action $y\in\actionA$ and a well-centered initial point $\h_0\in B_2(\polytope_y,-\frac{\eta}{2})$. Suppose the oracle $\approxmem(\eps_1,\eps_2,\eps_3)$ satisfies $\epsrobust\sqrt{\numP}\le \eps_2\le\frac{\eps_1}{\sqrt{\numP}}\le C\cdot\left(\frac{\eta\cdot\epsopt}{m}\right)^{13} $, where $C$ is a universal constant.
    Then the \textsc{MembershipOpt} algorithm (Algorithms 12 and 13 from \citep{haghtalab2022learning} with subroutine \textsc{ConservativeBestResponse} replaced by our $\approxmem$) terminates in $O\left(\numP\log\left(\frac{\numP}{\eta\cdot\epsopt}\right)\right)$ oracle calls to $\approxmem$, and returns a strategy $\hat{\h}$ that satisfies
    \begin{enumerate}
        \item \textbf{Optimality:} $U_P(\hat{\h},y)\ge\max_{\hst\in \polytope_y}U_P(\hst,y)-\epsopt$;
        \item \textbf{Robustness:} $\hat{\h}\in B_2(\polytope_y,-\epsrobust\sqrt{\numP})$.
    \end{enumerate}

\end{lemma}

\begin{lemma}[Robustness check]
    \label[lemma]{lemma:robustness-check}
    With probability at least $1-\frac{1}{10T}$, the following holds simultaneously for all pairs $(\tilde{\h}_i,y_i)$ returned by $\textsc{MembershipOpt}$: if $\tilde{\h}_i\in B_2(\polytope_{y_i},-\epsrobust\sqrt{\numP})$, then the robustness check step must return \textsc{True}. On the other hand, if the robustness check step returns \textsc{True}, then $\tilde{\h}_i$ satisfies $\tilde{\h}_i\in B_2(\polytope_{y_i},-\epsrobust)$.
\end{lemma}
\proof{Proof of \Cref{lemma:robustness-check}.}
By \Cref{lem:apx-memb-oracle}, $\approxmem$ either:
\begin{enumerate}
    \item Returns \textsc{True}, certifying that $\tilde{\h}_i\in B_2(\polytope_{y_i},-\epsrobust)$ with probability $1-\frac{1}{10T^2}$, or
    \item Returns \textsc{False}, certifying that $\tilde{\h}_i\notin B_2(\polytope_{y_i},-\epsrobust\frac{\sqrt{\numP}+1}{2})$.
\end{enumerate}

For the first claim: if $\tilde{\h}_i\in B_2(\polytope_{y_i},-\epsrobust\sqrt{\numP})$, then case 2 cannot occur since $B_2(\polytope_{y_i},-\epsrobust\sqrt{\numP})\subseteq B_2(\polytope_{y_i},-\epsrobust\frac{\sqrt{\numP}+1}{2})$, so $\approxmem$ must return \textsc{True}. For the second claim: if $\approxmem$ returns \textsc{True}, then case 1 applies and we have $\tilde{\h}_i\in B_2(\polytope_{y_i},-\epsrobust)$ with probability at least $1-\frac{1}{10T^2}$.  Taking a union bound over no more than $T$ pairs, we have the desired claim.
\Halmos
\endproof

Finally, we put the above lemmas together to prove \Cref{lemma:principal-main-complexity}.

\proof{Proof of \Cref{lemma:principal-main-complexity}.}
We begin by establishing the two guarantees for the output pair $(\hat{\h},\hat{y})$. The robustness guarantee $\Ptwo$ is a direct consequence of \Cref{lemma:robustness-check}. For the optimality guarantee $\Pone$, \Cref{lemma:initialization-set} guarantees that with probability at least $1-\frac{1}{5T}$, the initialization set $\calI$ includes a strategy $\h_0$ that lies within $B_2(\polytope_{\yst},-\frac{\eta}{2})$. By \Cref{lemma:optimization-phase}, when this strategy is used as the initial point for the optimization phase, it ensures that the solution set $\calS$ contains a strategy $\hat{\h}_0$ within $B_2(\polytope_{\yst},-\epsrobust\sqrt{\numP})$. This strategy achieves near-optimal utility $U_P(\hat{\h}_0,\yst)\ge\max_{\hst\in \polytope_{\yst}}U_P(\hst,\yst)-\epsopt=\Vst-\epsopt$ and passes the robustness check as per \Cref{lemma:robustness-check}. Consequently, since $(\hat{\h},\hat{y})$ is selected as the pair with the highest utility in $\calS$, it follows that:
\[
U_P(\hat{\h},\hat{y})\ge U_P(\hat{\h}_0,\yst)\ge\Vst-\epsopt.
\]
The total failure probability of all the membership calls across the optimization phase and the robustness check step is at most $T\cdot \frac{1}{10T^2}=\frac{1}{10T}$. Together with the failure probability of the initialization phase, Algorithm~\ref{algo:principal-main} succeeds with probability at least $1-T^{-1}$.

Now we calculate the sample complexity. For each pair $(\h_i,y_i)$, where $i\in[\Ninit]$, the number of rounds needed is:
\begin{itemize}
    \item The initialization phase makes one oracle call to $\approxmem(\eps_1=\frac{\eta}{4},\eps_2=\frac{\eta}{4\sqrt{\numP}},\eps_3=\frac{1}{10T^2})$, which takes $O\left(\sqrt{\numP} \cdot r_\delta^{-1}\left(\frac{\eta}{4\numA{\numP}}\right)\log T\right)$ rounds;
    \item By \Cref{lemma:optimization-phase}, the optimization phase makes $O\left(\numP\log\left(\frac{\numP}{\eta\cdot\epsopt}\right)\right)$ oracle calls to $\approxmem(\eps_1=C\sqrt{\numP}\left(\frac{\eta\cdot\epsopt}{m}\right)^{13},\eps_2=\epsrobust\sqrt{\numP},\eps_3=\frac{1}{10T^2})$, where each oracle call can be implemented in $O\left(\sqrt{\numP} \cdot r_\delta^{-1}\left(\frac{\epsrobust}{\numA}\right)\log T\right)$ rounds;
    \item The robustness check step makes one oracle call to $\approxmem(\eps_1=\frac{\epsrobust\sqrt{\numP}}{2},\eps_2=\frac{\epsrobust}{2},\eps_3=\frac{1}{10T^2})$, which takes $O\left(\sqrt{\numP} \cdot r_\delta^{-1}\left(\frac{\epsrobust}{2\numA\sqrt{\numP}}\right)\log T\right)$ rounds.
\end{itemize}
 
 Therefore, the total sample complexity is:
 \begin{align*}
  & \textstyle O\left(\Ninit\cdot\sqrt{\numP}\log(T) \cdot \left( r_\delta^{-1}\left(\frac{\eta}{4\numA{\numP}}\right)+ r_\delta^{-1}\left(\frac{\epsrobust}{2\numA\sqrt{\numP}}\right)\right)\right)
    \\
    =& O\left(\frac{\sqrt{\numP}}{\vol(\eta/2)} r_\delta^{-1}\left(\min\{\frac{\eta}{4\numA{\numP}},\frac{\epsrobust}{2\numA\sqrt{\numP}}\}\right)\log^2(T)\right).
 \end{align*}
 The proof is complete.
\Halmos
\endproof

\subsection{Proof of \Cref{thm:ETC-regr}}
\label{app:proof-etcreg}

In this section, we prove the main theorem (\Cref{thm:ETC-regr}) in \Cref{sec:learning-algos}.

\thmETCregr*

\begin{remark}[Adaptive regret versus calibration]
    Our primary focus lies on calibration due to its characterization of agents' beliefs and the fact that it provides both upper and lower bounds to the principal's utility. This is particularly useful for the learning direction, as denoted by the lower bounds in \Cref{thm:ETC-regr}. However, a different form of adaptive guarantee would suffice here: one concerning (external) regret. Nevertheless, we do not focus on regret as a characterization as it doesn't offer the same upper bound guarantees --- in fact, the principal could potentially extract more utility than $\Vst$. Additionally, regret-based assumptions tend to overly emphasize the agent's optimization techniques rather than maintaining a consistent belief about the action being executed.
\end{remark}

\proof{Proof of \Cref{thm:ETC-regr}.}
According to the round complexity of the exploration phase from \Cref{lemma:principal-main-complexity}, with prob.\ $\ge 1-T^{-1}$, the principal's cumulative regret can be bounded as
\begin{align*}
  &\Vst T-\Ex\Big[\sum_{t=1}^T U_P(h_t,y_t)\Big]\\
\;\lesssim\;&
{\frac{\sqrt{\numP}}{\vol(\eta/2)}\,r_\delta^{-1}\!\big(\tfrac{\epsrobust}{2\numA\sqrt{\numP}}\big)\log T}+\epsopt T
+\frac{\numA\sqrt{\numP}\,T\,r_\delta(T)}{\epsrobust};\\
\intertext{
With $r_\delta(t)=\Theta(t^{-1/\beta})$ we have $r_\delta(T)=\Theta(T^{-1/\beta})$ and 
$r_\delta^{-1}(x)=\Theta(x^{-\beta})$. The regret bound can be simplified as}
\lesssim\;&\frac{\numA^{\beta}\,\numP^{(\beta+1)/2}}{\vol(\eta/2)}\,\epsrobust^{-\beta} \log T 
+\epsopt T + \frac{k m^{1/2} T^{1-{1}/{\beta}}}{\epsrobust};\\
\intertext{To satisfy the constraint from \Cref{lemma:principal-main-complexity}, we set
$\epsrobust = \frac{C}{\sqrt{\numP}}\!\left(\frac{\eta\,\epsopt}{\numP}\right)^{13}$ and get}
\lesssim\;&
\frac{k^{\beta}\,m^{(\beta+1)/2}}{\vol(\eta/2)}\,
\Big(\frac{\sqrt{m}}{C}\Big)^{\beta}
\Big(\frac{m}{\eta\,\epsopt}\Big)^{13\beta}\!\log T
\;+\;\epsopt T
\;+\;\frac{k\,\sqrt{m}\,T^{1-1/\beta}}{\frac{C}{\sqrt{m}}\left(\frac{\eta\,\epsopt}{m}\right)^{13}}\\[2pt]
\lesssim\;&
\frac{k^{\beta}m^{14\beta+\frac12}}{\vol(\eta/2) \eta^{13\beta}}\;\epsopt^{-13\beta}\,\log T
\;+\;\epsopt T
\;+\;\numA\,m^{14}\,\eta^{-13}\,T^{1-1/\beta}\;\epsopt^{-13}.
\end{align*}

Balancing the last two terms yields 
\[
\epsopt^\star \;=\; \big(\numA\,m^{14}\,\eta^{-13}\,T^{-1/\beta}\big)^{1/14}
=\numA^{\frac{1}{14}}\,m\,\eta^{-\frac{13}{14}}\,T^{-\frac{1}{14\beta}}.
\]
At this choice, the regret bound becomes
\begin{align*}
    \lesssim\;& \numA^{\frac{1}{14}}\,m\,\eta^{-\frac{13}{14}}\,T^{\,1-\tfrac{1}{14\beta}}
\;+\;\frac{k^{\frac{\beta}{14}}\,m^{\beta+\frac{1}{2}} }{\vol(\eta/2)}\eta^{-\frac{13\beta}{14}}\,T^{\frac{13}{14}}\log T.
\end{align*}
The proof completes by taking the average regret.
\Halmos
\endproof

%% file: app-adaptive-calibration.tex
\section{Supplementary Material for \Cref{sec:algos-adaptive-calibration}}\label[section]{app:adaptive-calibration}

\subsection{Background on Sleeping Experts and \textsc{AdaNormalHedge}}

We start the exposition of this part by introducing the sleeping experts problem \citep{blum2007external,freund1997using}. For each expert $i\in[N]$ and round $t\in[T]$, let $\ell_{t,i}\in[0,1]$ be the loss of expert $i$, and let $I_{t,i}$ be an indicator that takes value $I_{t,i}=1$ if expert $i$ is active at round $t$ and $I_{t,i}=0$ if asleep. The interaction protocol at each round $t$ goes as follows: The indicators $\left(I_{t,i}\right)_{i\in[N]}$ are revealed to the learner. The learner selects a probability distribution $\pi_t\in\Delta([N])$ that is supported only on the set of active experts $A_t\triangleq\{i:I_{t,i}=1\}$. The adversary selects a loss vector $\left(\ell_{t,i}\right)_{i\in[N]}$. The learner then suffers expected loss $\hat{\ell}_t=\Ex_{i\sim \pi_t}\left[\ell_{t,i}\right]$. 
The regret with respect to each expert $i$ only accounts for the rounds when $i$ is awake, which, together with the fact that $\pi_t$ is only supported on active experts, implies that
\begin{equation}
    \reg_T(i) = \sum_{t \in [T]} I_{t,i} \left(\hat{\ell}_t-\ell_{t,i}\right)\quad\Rightarrow\quad \reg_T=\max_{i}\reg_T(i)\label{eq:sleeping-regr}
\end{equation}
One of the algorithms that can be used to provide sublinear regret for the sleeping experts problem is $\textsc{AdaNormalHedge}$~\citep{luo2015achieving}. $\textsc{AdaNormalHedge}$ is a powerful, parameter-free algorithm which provides regret bounds in terms of the cumulative magnitude of the \emph{instantaneous} regrets, defined as: $r_{t,i} = \hat{\ell}_{t} - \ell_{t,i}$ for all experts $i \in [N]$. {As its name suggests, $\textsc{AdaNormalHedge}$ uses the well-known algorithm $\textsc{Hedge}$ as a backbone; $\textsc{Hedge}$ maintains a probability distribution over experts at each round $t$ and draws an expert from said distribution. After the expert's loss is revealed, the probability distribution for the next round $t+1$ is updated using a multiplicative weights argument. For bandit feedback (i.e., when only the chosen expert's loss is revealed to the learner), the multiplicative weights update rule uses an inverse propensity scoring estimator for each expert's loss in place of their real loss. The new element that \textsc{AdaNormalHedge} brings to the table is a way of defining the weights at each round $t$; specifically, the weights are updated proportionally to the sum of instantaneous regret for each expert until round $t$. This allows the learner to obtain finer control over the total regret without needing extra parameters to tune the algorithm at each round.} The exact regret guarantee that \textsc{AdaNormalHedge} obtains is stated formally below.

\begin{lemma}[\textsc{AdaNormalHedge}~\citep{luo2015achieving}]
    Let $r_{t,i}=I_{t,i}\left(\hat{\ell}_t-\ell_{t,i}\right)$ be the instantaneous regret of any active expert $i\in A_t$ at round $t$, and $c_{t,i}=|r_{t,i}|$. Then,  \textsc{AdaNormalHedge} with prior $q\in\Delta([N])$ selects experts according to the following distribution
    \begin{align*}
        \pi_{t,i} &\propto q_i I_{t,i} w(R_{t-1,i},C_{t-1,i}), \text{ where }\\
        R_{t-1,i} &= \sum_{\tau \in [t-1]} r_{\tau,i},\ C_{t-1,i}=\sum_{\tau \in [t-1]} c_{\tau,i}, \\
        w(R,C) &= \frac{1}{2}\left(\Phi(R+1,C+1)-\Phi(R-1,C+1)\right), \\
        \Phi(R,C) &= \exp\left(\frac{\max\{0,R\}^2}{3C}\right)
    \end{align*}
    The regret of $\textsc{AdaNormalHedge}$ against any distribution over experts $u\in\Delta([N])$ is bounded by
    \begin{align*}
        \reg_T(u)\le O\left(\sqrt{\langle u, C_{T}\rangle \cdot\left(D_{\text{KL}}(u\|q)+\log\log T+\log\log N)\right)}\right).
    \end{align*}
    where by $D_{KL}(u \| q)$ we denote the KL-divergence between distributions $u$ and $q$.
\end{lemma}

\textsc{AdaNormalHedge} can be used to obtain adaptive regret bounds by creating a sleeping expert $(i,s)$ for each $i\in[N],s\in[T]$ that has the same loss as expert $i$ but is only awake after $s$.

\begin{corollary}
    Running \textsc{AdaNormalHedge} for the sleeping expert setting with prior {$q_{(i,s)}\propto\frac{1}{s^2}$} gives regret
    \begin{align*}
        \reg_t((i,s))\le O\left(\sqrt{(t-s)\left(\log (Ns)+\log\log T\right)}\right),
    \end{align*}
    where $T_i=\sum_{t=1}^T I_{t,i}$ is the total number of rounds in which $i$ is active.
\end{corollary}

\subsection{Formula for Computing $Q_t$ when $m=2$}

To obtain the explicit formula for $Q_t$, we first discretize the space of forecasts $\calF_P = [0,1]$ (since we focus on the case where $\numP = 2$) to form set $\calF_P^{\eps} = \{0, \eps, 2 \eps, \dots, 1- \eps, 1\}$. Then, we have that for each $\hat{\vp} \in \calF_P^\eps$: 
\begin{align*}
    \E_{g \sim \pi_t} \left[ \ell_{t,g}\right] &= \E_{g \sim \pi_t} \left[ L_g \left( \h_t, \vp \right)\right] = \sum_{g \in A_t(\mathcal{G})} \pi_{t,g} w_i(\vp) \sigma \left( \h_t - \vp\right) \\ 
    &= (\h_t - \vp) \underbrace{\sum_{i \in \calA_A}  w_i(\vp) \sum_{s\le t} \left( \pi_{t, g(s, i, +1)} - \pi_{t, g(s, i, -1)}\right)}_{Z_{\vp}} \numberthis{\label{eq:def-Zp}}
\end{align*}
{where we have omitted index $j$ from the sleeping expert $g$ since because $m = 2$, we can focus on a single action $j$.} {We assume WLOG that no forecast $\bp\in\calF_P^\eps$ falls exactly on the boundary of best response polytopes, so there is no tie-breaking needed.} From Equation~\eqref{eq:def-Zp}, we have that: 
\begin{equation}\label{eq:maxx}
    \max_{\h_t \in \calH_P} \E_{g \sim \pi_t} \left[ \ell_{t,g} \right] =\max\{Z_\vp , 0 \} - \vp Z_\vp
\end{equation}
where the equation also uses the fact that for $m = 2$, $\max_{\h_t \in \calH_P} {\h_t} = 1$.

In the final step, we map $\vp$ to the discretized grid of $\calF_P^\eps$. Let $j\eps, (j+1) \eps$ be two adjacent discretized points and $q \in [0,1]$ such that: $q Z_{j \eps} + (1 - q) Z_{(j+1) \eps} 0$. Then, setting $q_{t,j\eps} = q$ and $q_{t,(j+1)\eps} = 1-q$ and using Equation~\eqref{eq:maxx} gives that $$\max_{\h_t \in \calH_P}\Ex_{g\sim \pi_t \atop\vp\sim Q_t}\left[\ell_{t,g}\right]\le \eps.$$

%% file: app-continuous.tex
\section{Supplementary Material for \Cref{sec:continuous}}\label[section]{app:continuous}

\subsection{Proof of \Cref{thm:main-continuous}}
\label{app:continuous-main-proof}

\maincontinuous*
\subsubsection{Proof of Lower Bound}
\begin{proof}
Before delving into the proof of the lower bound, we first introduce some notations. Let $C(\h)\triangleq U_P(\h,\BR(\h))$.
Let $\Vst_\delta\triangleq\max_{\h\in B_2(\actionP,-\delta)}C(\h)$ be the optimal utility restricted in the smaller strategy set $B_2(\actionP,-\delta)$.
We use $\overline{y}_\phi\triangleq \frac{1}{M}\sum_{s\in[M]}y_{\phi,s}$ to denote the average feedback that $\lazygdwog$ uses to update the strategies.

We first consider any fixed $\eps>0$. Combining the guarantees of \Cref{lemma:feedback-close,lemma:utility-close}, we conclude that there exists a finite binning $\Pi_0$ and $M_{\eps}<\infty$, such that if the agent is $(0,\Pi_0)$-adaptively calibrated, then $\forall M\ge M_{\eps}$, the following two inequalities are satisfied at the same time:

    \begin{align}
        &\sup_{\phi\in[\Phi]}\left\|\overline{y}_{\phi}-\BR(\h_\phi)\right\|_2\le\eps&
        \text{(by \Cref{lemma:feedback-close})}\label{eq:guarantee-feedback}\\
        &\sup_{\phi\in\Phi}\frac{1}{M}\sum_{s\in[M]}U_P(\h_\phi,y_{\phi,s})\ge C(\h_\phi)-\eps;&\text{(by \Cref{lemma:utility-close})}\label{eq:guarantee-utility}
    \end{align}

    Set the parameters according to $\gamma_\phi=\gamma_0\numP^{-\frac{1}{2}}\phi^{-\frac{3}{4}}$ and $\delta_\phi\equiv\delta=\delta_0 \numP^{\frac{1}{2}}\Phi^{-\frac{1}{4}}$ in Algorithm~\ref{algo:gd}, then similar arguments to \citep[Theorem 3.1]{zrnic2021leads} guarantee that
    \begin{align*}
        \Vst_\delta-\frac{1}{\Phi}\sum_{\phi\in[\Phi]}\E[C(\h_\phi)]
        \le&
        \left(\frac{D_P^2}{2\gamma_0}+\frac{2W_P^2}{\delta_0^2}\right)\sqrt{\numP} \Phi^{-\frac{1}{4}}+ L_{\BR}D_P \frac{1}{\Phi}\sum_{\phi\in[\Phi]}\|\overline{y}_\phi-\BR(\h_\phi)\|_2\\
        \stepa{\le}& \left(\frac{D_P^2}{2\gamma_0}+\frac{2W_P^2}{\delta_0^2}\right)\sqrt{\numP} \Phi^{-\frac{1}{4}}+ L_{\BR}D_P\cdot\eps,
    \end{align*}
    where (a) is from \Cref{eq:guarantee-feedback}. 
    
    Now we upper bound the difference between $\Vst$ and $\Vst_\delta=\max_{\h\in B_2(\actionP,-\delta)}C(\h)$, then
    we have
    \[
        \Vst-\Vst_\delta\le
        \max_{\h^\star\in \actionP}\min_{\h'\in B_2(\actionP,-\delta)} C(\h^\star)-C(\h')
        \le L_U\max_{\h^\star\in \actionP}\min_{\h'\in B_2(\actionP,-\delta)} \|\h^\star-\h'\|_2 \le L_U\delta,
    \]
    where the second inequality follows from Assumption~\ref{assump:continuous} that $C(\h)$ is $L_U$-Lipschitz.

    The next step is to upper bound the difference between the actual average utility and $\frac{1}{\Phi}\sum_{\phi\in[\Phi]}\E[C(\h_\phi)]$. From \Cref{eq:guarantee-utility}, we have
    \[
        \frac{1}{\Phi}\sum_{\phi\in[\Phi]}\E[C(\h_\phi)]-\frac{1}{\Phi\epochlen}\sum_{\phi\in[\Phi]}\sum_{i\in[\epochlen]}U_P(\h_\phi,y_{\phi,i})\le\eps.
    \]
    Finally, putting the above inequalities together, we obtain
    \begin{align*}
        &\Vst-\frac{1}{\Phi\epochlen}\sum_{\phi\in[\Phi]}\sum_{i\in[\epochlen]}U_P(\h_\phi,y_{\phi,i})\\\le &
        \left(\Vst-\Vst_\delta\right)
        +\left(\Vst_\delta-\frac{1}{\Phi}\sum_{\phi\in[\Phi]}\E[C(\h_\phi)]\right)\\
        &\qquad\qquad+\left(\frac{1}{\Phi}\sum_{\phi\in[\Phi]}\E[C(\h_\phi)]-\frac{1}{\Phi\epochlen}\sum_{\phi\in[\Phi]}\sum_{i\in[\epochlen]}U_P(\h_\phi,y_{\phi,i})\right)\\
        \le& L_U \delta_0 \numP^{\frac{1}{2}}\Phi^{-\frac{1}{4}}+ \left(\frac{D_P^2}{2\gamma_0}+\frac{2W_P^2}{\delta_0^2}\right)\sqrt{\numP} \Phi^{-\frac{1}{4}}+ L_{\BR}D_P\cdot\eps +\eps.
    \end{align*}
    Taking the limit of $\Phi\to\infty$, the above inequalities imply 
    \[
        \lim_{\Phi\to\infty\atop \epochlen\to\infty}\frac{1}{\Phi\epochlen}\sum_{\phi\in[\Phi]}\sum_{i\in[\epochlen]}U_P(\h_\phi,y_{\phi,i})\ge V^\star-\eps\left(L_{\BR}D_P+1\right).
    \]
    Since the above arguments hold for all $\eps>0$, taking $\eps=\frac{\eps_0}{L_{\BR}D_P+1}$ proves the theorem.
\end{proof}

\subsubsection{Proof of Upper Bound}

\begin{proof}
For a fixed $\eps>0$, let $D_{\eps}=\{x_1,\cdots,x_I\}$ be an $\eps$-grid of $\calF_P$ under $\ell_2$ distance, and let $\Pi_0$ be the continuous binning specified by \Cref{eq:continuous-binning-construction}. We have:
\begin{align*}
    \sum_{t\in[T]} U_P(\h_t,y_t)=&
    \sum_{i\in[I]}\sum_{t\in[T]} w_i(\bp_t) U_P(\h_t,\BR(\bp_t))\\
    \stepa{\le}& \sum_{i\in[I]}\sum_{t\in[T]} w_i(\bp_t) \left(U_P(\h_t,\BR(x_i)) +L_2\cdot L_{\BR}\underbrace{\|\bp_t-x_i\|_2}_{\le2\eps}\right)\\
    \stepb{\le}& \sum_{i\in[I]}\left(\sum_{t\in[T]} w_i(\bp_t)\right) U_P\big(\frac{\sum_{t\in[T]} w_i(\bp_t) \h_t}{\sum_{t\in[T]} w_i(\bp_t)},\BR(x_i)\big)+2L_2 L_{\BR}\eps T\\
    \stepc{=}&\sum_{i\in[I]} n_T(i) U_P(\overline{\h}_T(i),\BR(x_i))+2L_2 L_{\BR}\eps T\\
    \stepd{\le}&\sum_{i\in[I]} n_T(i) \Big(U_P(\overline{\bp}_T(i),\BR(x_i))
    +L_1\left\|\overline{\bp}_T(i)-\overline{\h}_T(i)\right\|_2\Big)
    +2L_2 L_{\BR}\eps T\\
    =& \underbrace{\sum_{i\in[I]} n_T(i) U_P(\overline{\bp}_T(i),\BR(x_i))}_{(A)}
    +\underbrace{L_1\sum_{i\in[I]}n_T(i)\left\|\overline{\bp}_T(i)-\overline{\h}_T(i)\right\|_2}_{(B)}
    +2L_2 L_{\BR}\eps T\numberthis\label{eq:tmp-middle}
\end{align*}

In the above inequalities that lead to \eqref{eq:tmp-middle}, step (a) is because $U_P$ is $L_2$-Lipschitz in the second argument and $\BR(\cdot)$ is $L_{\BR}$-Lipschitz, and the fact that $w_i(\bp_t)>0$ only when $\|\bp_t-x_i\|_2<2\eps$. In step (b), we used Jensen's inequality because $U_P$ in concave in the first argument. Step (c) follows from the definition of $n_T(i)$ and $\overline{\h}_T(i)$ in \Cref{def:calibration-adaptive}. The last inequality (d) uses the fact that $U_P$ is $L_1$-Lipschitz in the first argument to decompose $U_P(\overline{\bp}_T(i),\BR(x_i))$ into calibration error (i.e., term (B)) and $U_P(\overline{\bp}_T(i),\BR(x_i))$ where the strategy that the agent best responds to is close to the principal's strategy (i.e., term (A)).

We can further bound $(A)$ and $(B)$ in Equation~\eqref{eq:tmp-middle} respectively as follows:
\[
(A)\le \sum_{i\in[I]} n_T(i) \left(U_P(x_i,\BR(x_i))+L_1\|x_i-\overline{\bp}_T(i)\|_2\right)\le V^\star T+ L_1 (2\eps) T,
\]
and
\[
(B)\le L_1 T\sum_{i\in[I]}\calerr_{i}(\h_{1:T},\bp_{1:T})
\le L_1 |D_\eps| r_\delta(T) T\quad \text{w.p. }\ge1-\delta. 
\]
Therefore, putting the above bounds together, we obtain that with probability $\ge1-\delta,$
\begin{align*}
    \frac{1}{T}\sum_{t\in[T]} U_P(\h_t,y_t) \le \Vst + (L_1 |D_\eps|)r_\delta(T)
    +2(L_1+L_2 L_{\BR})\eps.
\end{align*}
Since the above derivation holds for any $\eps>0$, it suffices to take $\eps$ such that $2(L_1+L_2 L_{\BR})\eps=\eps_0$. Finally, since $|D_\eps|<\infty$ and $r_\delta(T)=o(1)$, taking the limit of $T\to\infty$ proves the upper bound:
\[
    \lim_{T\to\infty}\frac{1}{T}\sum_{t\in[T]} U_P(\h_t,y_t) \le \Vst+\eps_0.
\]
\end{proof}

\subsection{Key lemma: asymptotically correct forecast}
\label{app:continuous-key-lemma}

In this section, we state and prove the key lemma for establishing \Cref{thm:main-continuous}. Intuitively, this lemma states that for any strategy $\h\in\actionP$, as long as the principal repeatedly plays $\h$ for enough rounds, the fraction of times where the agent's forecast is close to $\h$ will converge to 1.
\begin{lemma}
\label[lemma]{lemma:continuous-repeated}
    For any $\eps_0>0$, there exists a finite binning $\Pi_0$, such that if the principal repeatedly plays any $\h\in\actionP$ for $\epochlen$ rounds and the agent's forecasts $\bp_{1:\epochlen}$ are $(0, \Pi_0)$- adaptively calibrated, then:
    \begin{align}
\lim_{\epochlen\to\infty}\quad\frac{1}{\epochlen}\Big|\{s\in[\epochlen]:\|\bp_s-\h\|_2\ge\eps_0\}\Big|=0
        \label{eq:convergence-only}
    \end{align}
    In particular, if the calibration error (defined in \Cref{def:calibration-adaptive}) has rate $r(\cdot)\in o(1)$ with respect to $\Pi_0$, then 
    \begin{align}
        \frac{1}{\epochlen}\Big|\{s\in[\epochlen]:\|\bp_s-\h\|_2\ge\eps_0\}\Big|
        \le\frac{8\sqrt{\numP}|\Pi_0|^2}{\eps_0}r(\epochlen).
        \label{eq:rate}
    \end{align}
\end{lemma}

\proof{Proof of \Cref{lemma:continuous-repeated}.}
    We first describe the construction of $\Pi_0$.
    For $\eps=\frac{1}{4}\eps_0$, let $D_{\eps}=\{x_1,\cdots,x_I\}$ be an $\eps$-grid of $\calF_P$ under $\ell_2$ distance, and $\Lambda(\bp;x,R)\triangleq\left(R-\|\bp-x\|_2\right)_+$ be the tent function with center $x$ and radius $R$. Consider the following binning 
    \begin{align}
        \Pi_0=\left\{w_i(\bp)\triangleq\frac{\Lambda(\bp;x_i,2\eps)}{\sum_{j\in[I]}\Lambda(\bp;x_j,2\eps)}:x_i\in D_{\eps}\right\}.
        \label{eq:continuous-binning-construction}
    \end{align}
    Clearly, $|\Pi_0|=I<\infty$ because the diameter of $\calF_P$ is bounded as stated in Assumption~\ref{assump:continuous}. We can also verify that $\Pi_0$ satisfies $\sum_{i\in[I]}w_i(\bp)=1$ for all $\bp\in\calF_P$ because $w_i(\bp)$ is defined as the normalized tent function. 

    Now we prove that $\Pi_0$ satisfies the desired property.
    Since the agent is adaptively calibrated to $\Pi_0$, we have that $\forall i\in[I]$,
    \begin{align*}
        \frac{n_{[M]}(i)}{M}\left\|\bar{\bp}_{[M]}(i)-\h\right\|_2
        \le\sqrt{\numP} \lim_{M\to\infty}\frac{n_{[M]}(i)}{M}\left\|\bar{\bp}_{[M]}(i)-\h\right\|_\infty\le\sqrt{m}r(M).
    \end{align*}
    
    Now, for $\delta=3\eps=\frac{3}{4}\eps_0$, let $D_\eps^{(\delta)}\subseteq D_\eps$ be defined as
    \begin{align}
        D_\eps^{(\delta)}=\left\{x_i\in D_\eps:\|x_i-\h\|\ge\delta\right\}.
        \label{eq:def-D-eps-delta}
    \end{align}
    Since $|D_\eps^{(\delta)}|\le|D_\eps|=I<\infty$, taking the sum of calibration error over bins in $D_\eps^{(\delta)}$, we obtain
    \begin{align}
        \sum_{x_i\in D_\eps^{(\delta)}}\frac{n_{[M]}(i)}{M}\left\|\bar{\bp}_{[M]}(i)-\h\right\|_2
        =\frac{1}{\epochlen}\sum_{x_i \in D_\eps^{(\delta)}}\left\|
        \sum_{s\in[\epochlen]}w_i(\bp_s)(\h-\bp_s)
        \right\|_2\le \sqrt{\numP}I r(\epochlen).\label{eq:tmp}
    \end{align}
    We can further lower bound \eqref{eq:tmp} and get:
    \begin{align*}
        &\frac{1}{\epochlen}\sum_{x_i \in D_\eps^{(\delta)}}\left\|
        \sum_{s\in[\epochlen]}w_i(\bp_s)(\h-\bp_s)
        \right\|_2\\
        =&\;\frac{1}{\epochlen}\sum_{x_i \in D_\eps^{(\delta)}}\left\|
        \sum_{s\in[\epochlen]}w_i(\bp_s)\Big((\h-x_i)+(x_i-\bp_s)\Big)
        \right\|_2\\
        \stepa{\ge}&\;\frac{1}{\epochlen}\sum_{x_i \in D_\eps^{(\delta)}}\left(\left\|
        \sum_{s\in[\epochlen]}w_i(\bp_s)(\h-x_i)
        \right\|_2-\left\|
        \sum_{s\in[\epochlen]}w_i(\bp_s)(x_i-\bp_s)
        \right\|_2\right)\\
        \stepb{\ge}&\;\frac{1}{\epochlen}\sum_{x_i \in D_\eps^{(\delta)}}\sum_{s\in[\epochlen]}w_i(\bp_s)\Big(\left\|
        \h-x_i
        \right\|_2-\left\|
        x_i-\bp_s
        \right\|_2\Big)\\
        \stepc{\ge}&\;\frac{1}{\epochlen}\sum_{x_i \in D_\eps^{(\delta)}} n_{[\epochlen]}(i)(\delta-2\eps)\ge\frac{\eps_0}{4\epochlen}\sum_{x_i \in D_\eps^{(\delta)}} n_{[\epochlen]}(i).
    \end{align*}
    In the above inequalities, (a) and (b) are both due to triangle inequalities, and (c) is because $\|\h-x_i\|_2\ge\delta$ from the definition of $D_\eps^{(\delta)}$ in \eqref{eq:def-D-eps-delta} and $\|x_i-\bp_s\|_2<2\eps$ whenever $w_i(\bp_s)>0\iff \Lambda(\bp_s;x_i,2\eps)>0$. Together with \eqref{eq:tmp}, the above set of inequalities imply
    \begin{align}
        \frac{1}{\epochlen}\sum_{x_i \in D_\eps^{(\delta)}} n_{[t]}(i)\le\left(\frac{4}{\eps_0}\right)
         \frac{1}{\epochlen}\sum_{x_i \in D_\eps^{(\delta)}}\left\|
        \sum_{s\in[\epochlen]}w_i(\bp_s)(\h-\bp_s)
        \right\|_2\le \frac{4\sqrt{\numP}I}{\eps_0}r(\epochlen).
        \label{ineq:effective-count}
    \end{align}
    
    On the other hand, since $D_\eps$ is an $\eps$-grid of $\calF_P$, if $\|\bp_s-\h\|_2\ge\eps_0$, there must exist $x_i\in D_\eps$ such that $\|x_i-\bp_s\|_2\le\eps$, which implies 
    \begin{align*}
        \|x_i-\h\|_2\ge \|\bp_s-\h\|_2-\|x_i-\bp_s\|_2\ge\eps_0-\eps=\frac{3}{4}\eps_0=\delta\quad\Rightarrow\quad
        x_i\in D_\eps^{(\delta)}.
    \end{align*}
    As for the weight that $w_i$ assigns to $\bp_s$, we also have
    \begin{align*}
        w_i(\bp_s)=\frac{\Lambda(\bp_s;x_i,2\eps)}{\sum_{j\in[I]}\Lambda(\bp_s;x_j,2\eps)}\ge\frac{2\eps-\eps}{I\cdot 2\eps}=\frac{1}{2I}.
    \end{align*}
    Therefore, we have
    \begin{align}
        \frac{1}{\epochlen}\Big|\{s\in[\epochlen]:\|\bp_s-\h\|\ge\eps_0\}\Big|\le \frac{1}{\epochlen}\sum_{x_i \in D_\eps^{(\delta)}}\sum_{s\in[\epochlen]} (2I) w_i(\bp_s)=\frac{2I}{t}\sum_{x_i \in D_\eps^{(\delta)}} n_{[\epochlen]}(i)\label{eq:goal}
    \end{align}
    Finally, combining inequalities \eqref{ineq:effective-count} and \eqref{eq:goal}, we conclude that
    \begin{align*}
        \frac{1}{\epochlen}\Big|\{s\in[\epochlen]:\|\bp_s-\h\|\ge\eps_0\}\Big|
        \le (2I)\lim_{\epochlen\to\infty}\frac{1}{\epochlen}\sum_{x_i \in D_\eps^{(\delta)}} n_{[\epochlen]}(i)
        \le \frac{8\sqrt{\numP}I^2}{\eps_0}r(\epochlen),
    \end{align*}
    which proves \eqref{eq:rate}. The proof is complete by taking the limit of $\epochlen\to0$, which guarantees $r(\epochlen)\to0$ and immediately implies the convergence result in \eqref{eq:convergence-only}.\Halmos
\endproof

Note that the rate in Equation~\eqref{eq:rate} does not depend on strategy $\h$. Therefore, in the context of running $\lazygdwog$ (Algorithm~\ref{algo:gd}), we can turn \Cref{lemma:continuous-repeated} into the following uniform convergence result across epochs:

\begin{proposition}
\label{prop:avg-action-uniform-convergence}
For any $\eps_0>0$, there exists a finite binning $\Pi_0$, such that $\forall\Phi>0$, if 
the principal runs $\lazygdwog$ for $\Phi$ epochs where each epoch has length $\epochlen$, and
the agent's forecasts $(\bp_{\phi,s})_{\phi\in[\Phi],s\in[\epochlen]}$ are adaptively calibrated with respect to $\Pi_0$, then we have the following \emph{uniform convergence} guarantee:
    \begin{align}
        \lim_{\epochlen\to\infty}\sup_{\phi\in[\Phi]}
        \frac{1}{\epochlen}\Big|\{s\in[\epochlen]:\|\bp_{\phi,s}-\h_\phi\|_2\ge\eps_0\}\Big|=0
        \label{eq:uniform-convergence}
    \end{align}
\end{proposition}

\begin{remark}
    \label[remark]{remark:l1-calibration}
    Note that the rate in \eqref{eq:rate} has a polynomial dependency on $|\Pi_0|$, which, due to the construction in the proof of \Cref{lemma:continuous-repeated}, ends up being exponential in $\numP$ because it is the size of a $\frac{\eps_0}{4}$ grid of the domain $\actionP$. To improve on this exponential dependency, one possible approach is to design an adaptive calibration algorithm for the agent that achieves the stronger notion of $\ell_1$ calibration, which is more common in recent literature. For example, \citet{hart2022calibrated,foster1997calibrated,foster1998asymptotic} are defined using $\ell_1$ calibration error rather than $\ell_\infty$. Another approach is to avoid using naive conversion from $\ell_\infty$ to $\ell_1$ calibration error in \eqref{eq:tmp}, which leads to a polynomial dependency on the number of bins.
    These two approaches are equivalent ways of formulating the problem, and they both lead to interesting open directions.
\end{remark}

\subsection{More auxiliary lemmas: approximate best response and closeness in utility}

In this section, we use the results in \Cref{app:continuous-key-lemma} to show that the average feedback $\frac{1}{\epochlen}\sum_{s\in[\epochlen]}y_{\phi,s}$ in epoch $\phi\in[\Phi]$ is close to the best response $\BR(\h_{\phi})$ (\Cref{lemma:feedback-close}), and that the principal's average utility in this epoch is close to $U_P(\h_\phi,\BR(\h_{\phi}))$ (\Cref{lemma:utility-close}).

\begin{lemma} For any $\eps_1>0$, there exists a finite binning $\Pi_0$ and $M_0<\infty$ such that when agent's forecasts $\bp_{1:t}$ are adaptively calibrated with respect to $\Pi_0$, then we have that $\forall M\ge M_0$,
\label[lemma]{lemma:feedback-close}
    \begin{align*}
        \sup_{\phi\in[\Phi]}\left\|\frac{1}{\epochlen}\sum_{s\in[\epochlen]}y_{\phi,s}-\BR(\h_\phi)\right\|_2\le\eps_1.
    \end{align*}
\end{lemma}
\begin{proof}
Let $\eps_0=\frac{\eps_1}{2 L_{\BR}}$ and $\Pi_0$ be the binning that satisfies \Cref{prop:avg-action-uniform-convergence} for parameter $\eps_0$.
Therefore, we know from Equation~\eqref{eq:uniform-convergence} in \Cref{prop:avg-action-uniform-convergence} that for $\eps_2=\frac{\eps_1}{2\cdot D_P\cdot L_{\BR}}$
there exists $M_0$ such that $\forall M\ge M_0$,
\begin{align}
    \sup_{\phi\in[\Phi]}\frac{1}{M}\Big|\{s\in[M]:\|\bp_{\phi,s}-\h_\phi\|_2\ge\eps_0\}\Big|\le \eps_2.
    \label{eq:fraction-tmp}
\end{align}
Using Lipschitzness of the best response mapping $\BR(\cdot)$, we have that $\forall\phi\in[\Phi]$,
    \begin{align*}
        &\left\|\frac{1}{M}\sum_{s\in[M]}y_{\phi,s}-\BR(\h_\phi)\right\|\\
        \le& \frac{1}{M} \sum_{s\in[\epochlen]} \|y_{\phi,s}-\BR(\h_\phi)\|_2
        & \tag{Triangle inequality}
        \\
        \le& L_{\BR}\frac{1}{M}\sum_{s\in[M]}\left\|\bp_{\phi,s}-\h_\phi\right\|_2
        & \tag{$\BR(\cdot)$ is $L_{\BR}$-Lipschitz}\\
        \le& L_{\BR}\frac{1}{M}\left(
        \sum_{s\in[M]:\|\bp_{\phi,s}-\h_\phi\|\ge\eps_0}\diam{\cH_P}
        +\sum_{s\in[M]:\|\bp_{\phi,s}-\h_\phi\|<\eps_0}\eps_0
        \right)\\
        \le& L_{\BR}\frac{1}{M}\left( \eps_2 M\cdot D_P+M\cdot\eps_0  \right)
        &\tag{Eq.~\eqref{eq:fraction-tmp} \& $\diam{\cH_P}\le D_P$}
        \\
        \le& D_P\cdot L_{\BR}\cdot \eps_2+
         L_{\BR}\cdot\eps_0=\frac{\eps_1}{2}+\frac{\eps_1}{2}=\eps_1.
    \end{align*}
\end{proof}

\begin{lemma}
    \label[lemma]{lemma:utility-close}
    For any $\eps_1>0$, there exists a finite binning $\Pi_0$ and $M_0<\infty$ such that when agent's forecasts $\bp_{1:t}$ are adaptively calibrated with respect to $\Pi_0$, then we have that $\forall M\ge M_0$,
    \begin{align}
        \sup_{\phi\in[\Phi]}\left|\frac{1}{\epochlen}\sum_{s\in[\epochlen]}
        U_P(\h_\phi,y_{\phi,s})-U_P(\h_\phi,\BR(\h_\phi))\right|\le\eps_1.
    \end{align}
\end{lemma}
\begin{proof}
    The proof of this lemma is very similar to that of \Cref{lemma:feedback-close}, with a different choice of constants $\eps_0$ and $\eps_2$. Note that since $U_P$ is $L_2$-Lipschitz in the second argument, we have
    \begin{align*}
        \left|\frac{1}{\epochlen}\sum_{s\in[\epochlen]}
        U_P(\h_\phi,y_{\phi,s})-U_P(\h_\phi,\BR(\h_\phi))\right|
        \le&\frac{1}{M}\sum_{s\in[\epochlen]}\left\|U_P(\h_\phi,y_{\phi,s})-U_P(\h_\phi,\BR(\h_\phi))\right\|_2\\
        \le&L_2\cdot\frac{1}{M}\sum_{s\in[\epochlen]}\left\|y_{\phi,s}-\BR(\h_\phi)\right\|_2.
    \end{align*}
    The rest of the proof follows from \Cref{lemma:feedback-close} by choosing $\eps_0=\frac{\eps_1}{2 L_2 L_{\BR}}$ and $\eps_2=\frac{\eps_1}{2\cdot D_P\cdot L_{\BR}L_2}$.
\end{proof}